\documentclass[journal,onecolumn]{IEEEtran}
\usepackage{setspace}
\usepackage{amsmath}
\usepackage{amsthm}
\usepackage{amssymb}
\usepackage{amsfonts}
\usepackage{color}
\usepackage{graphicx}
\usepackage{epstopdf}
\usepackage{subfigure}  
\usepackage{cite}
\usepackage{mathtools}
\usepackage{stmaryrd}
\usepackage[mathscr]{euscript}
\usepackage{mathrsfs}
\newtheorem{prop}{Proposition}
\usepackage{soul}
\usepackage{arydshln}
\usepackage{mathtools, cuted}
\usepackage{stfloats}
\usepackage{xcolor}

\usepackage{url}
\usepackage[justification=raggedright]{caption} 
\allowdisplaybreaks

\newtheorem{theorem}{Theorem}
\newtheorem{lemma}{Lemma}

\ifCLASSINFOpdf

\ifCLASSINFOpdf
\else
\fi
\begin{document}
%
%
%
%


\title{Generalized Design of Sampling Kernels for 2-D FRI Signals}
%
\author{Saurav Kumaraswami Shastri, Sunil Rudresh, and Chandra Sekhar Seelamantula,~\IEEEmembership{Senior Member,~IEEE} \thanks {The authors are with the Department of Electrical Engineering, Indian Institute of Science, Bangalore - 560012, India. Email:\{sauravs, sunilr, chandrasekhar\}@iisc.ac.in}}

\maketitle
\begin{abstract}
One of the interesting problems in the finite-rate-of-innovation signal sampling framework is the design of compactly supported sampling kernels. In this paper, we present a generic framework for designing sampling kernels in 2-D. We consider both {\it separable} and {\it nonseparable} kernels. The design is carried out in the frequency domain, where a set of alias cancellation conditions are imposed on the kernel's frequency response. The Paley-Wiener theorem for 2-D signals is invoked to arrive at admissible kernels with a compact support. As a specific case, we show that a certain separable extension of the 1-D design framework results in 2-D sum-of-modulated-spline (SMS) kernels. Similar to their 1-D counterparts, the 2-D SMS kernels have the attractive feature of reproducing a class of 2-D polynomial-modulated exponentials of a desired order. Also, the support of the kernels is independent of the order. The design framework is generic and also allows one to design nonseparable sampling kernels. To this end, we demonstrate the design of a nonseparable kernel and present simulation results.
\end{abstract}


%
\IEEEpeerreviewmaketitle

\section{Introduction}
%
%
%
%
\IEEEPARstart{I}{n} their seminal work, Vetterli et al. \cite{vetterli} developed a sampling framework for a certain class of nonbandlimited signals that have a finite rate of innovation (FRI) or finite degrees of freedom per unit time/space. The sampling process consists of passing the signal through a suitable kernel, followed by sampling the resulting signal at specific instants. Under certain conditions, the samples thus obtained are sufficient to completely characterize the signal. The degrees of freedom or parameters of the signal are estimated using a suitable reconstruction technique that is coupled to the sampling process. One of the major aspects of the FRI sampling framework is the design of suitable sampling kernels that are realizable and applicable to a larger class of FRI signals. 

Consider the 2-D signal 
\begin{equation}
\label{eq:2-Dfri_time}
    f(x,y) = \sum_{\ell = 1}^{L}{\gamma_\ell}\,h(x - x_\ell, y - y_\ell),
\end{equation}
which is a sum of scaled and shifted versions of a known function $h(x,y)$, and $\{\gamma_\ell, x_\ell, y_\ell\}_{\ell=1}^{L}$ are the parameters that specify the signal $f(x,y)$. The signal in \eqref{eq:2-Dfri_time} is an FRI signal with $3L$ degrees of freedom. This signal model is frequently encountered in several { \it imaging applications} such as localization microscopy \cite{betzig2006imaging,folling2008fluorescence}, astronomical imaging \cite{ImageAstronomy_Molina, PantinEtAl07}, deflectometry \cite{sudhakar2013compressive}, etc. In this paper, we address the problem of designing suitable sampling kernels for signals of the form \eqref{eq:2-Dfri_time}.  Specifically, we develop a generalized kernel design methodology.  Before proceeding with the design framework, we briefly recall the various 1-D and 2-D sampling kernels proposed in the literature so far.

\subsection{Sampling Kernels}
In the 1-D case, Vetterli et al. \cite{vetterli} proposed infinitely supported sinc and Gaussian kernels, whereas Dragotti et al. \cite{fri_strang} designed a class of compactly supported kernels that reproduce polynomials or exponentials. Tur et al. \cite{eldar_sos} developed alias cancellation conditions and proposed the sum-of-sincs kernel in the frequency domain. Seelamantula and Unser \cite{css}, and Olkkonen and Olkkonen \cite{olkkonen} employed practiaclly realizable kernels derived from resistor-capacitor circuits to sample and reconstruct a stream of Diracs. Recently, Mulleti and Seelamantula \cite{satish_paley_wiener} developed a generalized method for designing 1-D sampling kernels based on the Paley-Wiener theorem, and as a specific construction, they focussed on the class of sum-of-modulated-spline (SMS) kernels. 

Following the 1-D sampling framework \cite{vetterli}, Maravi\'c and Vetterli \cite{maravic_2D} developed a sampling and reconstruction framework for 2-D FRI signals with 2-D sinc and Gaussian functions as sampling kernels. Shukla and Dragotti \cite{shukla2007sampling} considered multidimensional FRI signals and proposed polynomial reproducing kernels for sampling convex shapes and polygons. As an application to step-edge detection in images, Baboulaz et al. \cite{baboulaz} developed a local reconstruction scheme with the B-spline kernel. Improving upon \cite{baboulaz}, Hirabayashi et al. \cite{hirabayashi2010spline} designed exponential reproducing kernels and showed that they perform better than polynomial reproducing kernels. Chen et al. \cite{chen20122D} generalized the B-spline sampling kernels for step-edge detection to polygon signal reconstruction using practical sampling kernels. Pan et al. \cite{FRI_Curves_Pan} employed the 2-D sinc sampling kernel and developed a reconstruction scheme for a certain class of parameterizable 2-D curves. Depending upon the number of parameters, the signal model efficiently represents a wider class of curves that are more complex than polygons. Recently, De and Seelamantula \cite{anindita_2D_fri} developed the separable extension of the 1-D non-repeating sum-of-sincs (NR-SoS) kernel \cite{mulleti2014ultrasound} to arrive at 2-D NR-SoS and 2-D SMS kernels.  
  
\subsection{This Paper}
The main contribution of this paper is a generalized framework for designing compactly supported 2-D sampling kernels. The framework allows for the design of separable or nonseparable kernels. To the best of our knowledge, this is the first methodology for the generic design of nonseparable kernels. Starting from a set of alias cancellation conditions that have to be satisfied by a sampling kernel in the frequency domain, admissible sampling kernels with compact support are developed. The characterization of a kernel starting from alias cancellation conditions to enforcing compact support is based on the Paley-Wiener theorem for functions in $\mathbb{R}^d$, $d \geq 1$. We leverage the 1-D kernel design framework proposed in \cite{satish_paley_wiener} and extend it to 2-D. To begin with, we show that a certain separable extension of the 1-D case results in 2-D SMS kernels \cite{anindita_2D_fri}, which have the attractive feature of reproducing separable 2-D polynomial-modulated exponentials. Further, to demonstrate the generalizability of the proposed design framework, we develop a nonseparable sampling kernel. We present simulation results demonstrating exact recovery of a 2-D Dirac stream using one such nonseparable sampling kernel. 

\section{The 2-D FRI Sampling and Reconstruction Problems}
A schematic of the kernel-based sampling framework is shown in Fig.~\ref{fig:block}, where the input signal $f$ is passed through a suitable kernel $g$. The resulting signal $\psi(x,y)=(f*g)(x,y)$ is sampled with the sampling intervals $T_{sx}$ and $T_{sy}$ along $x$ and $y$ axes, respectively to get the measurements $\{\psi(n_1T_{sx},n_2T_{sy})\}$, $\{(n_1, n_2)\}\in \mathbb{Z}^2$.

Let $F(s_x,s_y)=\int_{\mathbb{R}^2}f(x,y) e^{-(xs_x+ys_y)} \mathrm{d}x \mathrm{d}y$ be the Laplace transform of $f$, where $s_x=\sigma_x+\mathrm{j}\Omega_x$ and $s_y=\sigma_y+\mathrm{j}\Omega_y$. The 2-D continuous-time Fourier transform (CTFT) of the signal $f$  in~\eqref{eq:2-Dfri_time} is 
\begin{equation}
    \label{eq:2-Dfri_freq}
    F(\mathrm{j}\Omega_x,\mathrm{j}\Omega_y) = H(\mathrm{j}\Omega_x,\mathrm{j}\Omega_y)\,\sum_{\ell = 1}^{L}\gamma_\ell e^{-\mathrm{j}(\Omega_x x_\ell + \Omega_y y_\ell)},
\end{equation}
where $H(\mathrm{j}\Omega_x,\mathrm{j}\Omega_y)$ is the 2-D CTFT of $h(x,y)$. Let  $\mathcal{S}$ denote the set of frequency locations in the $\Omega_x-\Omega_y$ plane defined as $\mathcal{S}=\{(k_1\Omega_{0x}, k_2\Omega_{0y})\}_{k_1\in\mathcal{K}_1, k_2\in\mathcal{K}_2}$ for some non-zero $\Omega_{0x}$ and $\Omega_{0y}$, where $\mathcal{K}_1$ and $\mathcal{K}_2$ are sets of contiguous integers chosen suitably based on the model order $L$ and the noise statistics. Throughout the paper, unless specified otherwise, $k_1 \in \mathcal{K}_1$ and $k_2 \in \mathcal{K}_2$. Now, the measurements of $P(\mathrm{j}\Omega_x,\mathrm{j}\Omega_y) = \frac{F(\mathrm{j}\Omega_x, \mathrm{j}\Omega_y)}{H(\mathrm{j}\Omega_x, \mathrm{j}\Omega_y)}$ evaluated on $\mathcal{S}$ are given by 
\begin{align}
 P(\mathrm{j}k_1\Omega_{0x},\mathrm{j}k_2\Omega_{0y}) =  \sum_{\ell = 1}^{L}\gamma_\ell e^{-\mathrm{j}(k_1\Omega_{0x} x_\ell + k_2\Omega_{0y} y_\ell)} 
 \label{eq:FRI_kwo}.
\end{align}
To avoid singularities, the set $\mathcal{S}$ is chosen such that $H(\mathrm{j}\Omega_x, \mathrm{j}\Omega_y) \neq 0$ on $\mathcal{S}$. The right-hand side of~\eqref{eq:FRI_kwo} is in the form of sum-of-weighted-complex exponentials (SWCEs). The estimation of the parameters $\{x_\ell, y_\ell\}_{\ell = 1}^{L}$ from the measurements of the SWCE form in~\eqref{eq:FRI_kwo} is performed by employing high-resolution spectral estimation (HRSE) techniques~\cite{moses} such as annihilating filter~\cite{prony} applied suitably for 2-D case or 2-D subspace methods~\cite{2DESPRIT_rouquette2001,2DHarmonicRetrieval_vanpoucke1994,2DparameterPaulraj1998,2DESPRIT_haardt1995}. For more details about the application of these techniques, conditions on the minimum number of measurements required for exact recovery in the absence of noise, etc., the reader is referred to~\cite{maravic_2D, maravic2005sampling}.  
\begin{figure}[t]
\centering
{\includegraphics[width=4.0in]{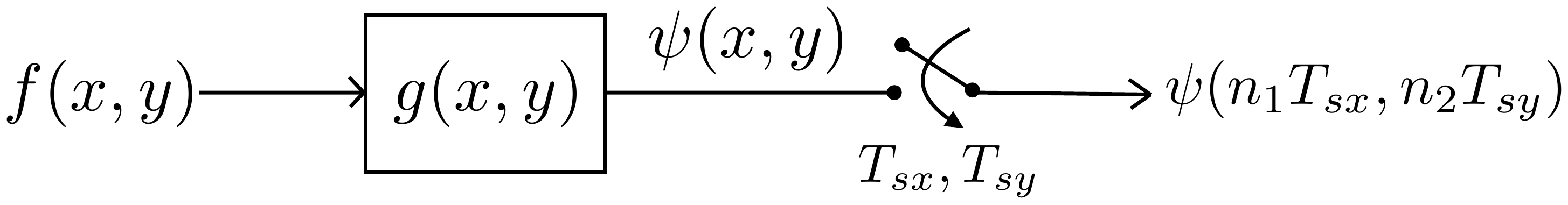}}
\caption{Schematic of a 2-D kernel-based sampling scheme.}
\label{fig:block}
\end{figure}

As $h$ is assumed to be known a priori, given the non-aliased samples of $F$ on $\mathcal{S}$, one could exactly recover the parameters $\{ x_\ell, y_\ell\}_{\ell = 1}^{L}$. The sampling kernel $g$ has to be designed in such a way that the non-aliased samples of $F$ on $\mathcal{S}$ can be obtained by the spatial-domain measurements $\{\psi(n_1T_{sx},n_2T_{sy})\}$.  Recently, we \cite{anindita_2D_fri} derived the conditions on the frequency response of the sampling kernel $g$, which are necessary to obtain non-alised samples of $F$ on $\mathcal{S}$:
\begin{align}
\centering
    &G(\mathrm{j}k_1\Omega_{0x},\mathrm{j}k_2\Omega_{0y}) \neq 0, \, \forall\, k_1\in\mathcal{K}_1, k_2\in\mathcal{K}_2,\,\text{and} \label{eq:alias_cancle_1} \\
    &G(\mathrm{j}k_1\Omega_{0x} + \mathrm{j}m_1\Omega_{sx}\, ,\, \mathrm{j}k_2\Omega_{0y} + \mathrm{j}m_2\Omega_{sy}) = 0,\,\quad  \forall\,k_1\in\mathcal{K}_1, k_2\in\mathcal{K}_2, \, \text{and}\, \forall\, m_1, m_2 \in\mathbb{Z}\backslash\{0\}, \label{eq:alias_cancle_2} \\
    & \text{with} \,\, \Omega_{sx}\geq |\mathcal{K}_1|\Omega_{0x}\,\, \text{and}\,\, \Omega_{sy}\geq |\mathcal{K}_2|\Omega_{0y}.
    \label{eq:alias_cancle_3} 
\end{align}
The spatial-domain sampling intervals are given by $T_{sx}=\frac{2\pi}{\Omega_{sx}}$ and $T_{sy}=\frac{2\pi}{\Omega_{sy}}$. Also, it was shown that if the kernel $g$ satisfies  (\ref{eq:alias_cancle_1})-(\ref{eq:alias_cancle_3}), then the 2-D discrete-time Fourier transform of $\{\psi(n_1T_{sx},n_2T_{sy})\}$ evaluated on $\mathcal{S}$ gives the non-aliased samples of $F$ on $\mathcal{S}$.

\section{Generalized Design of 2-D Sampling Kernels}
Recently, a generalized method for designing compactly supported sampling kernels for 1-D FRI signals was developed by Mulleti and Seelamantula \cite{satish_paley_wiener}. We extend that framework to the design of 2-D sampling kernels (both separable and nonseparable varieties). The framework consists of two steps: first, to design kernels that satisfy the alias-cancellation conditions (\ref{eq:alias_cancle_1})-(\ref{eq:alias_cancle_2}), and second, enforcing the kernels to be compactly supported. 

Consider a sampling kernel $g(x,y)$, that has a rational 2-D Laplace transform of the form
\begin{equation}
    \displaystyle G(s_x,s_y)= \frac{V(s_x,s_y)\,Z(s_x,s_y)}{U(s_x,s_y)},
    \label{eq:gen_samp_kernel}
\end{equation}
where $Z(s_x,s_y)$ is a function that has zeros at $\{n_1\Omega_{0x},n_2\Omega_{0y}\}$, $\forall\,\{(n_1,n_2)\}\, \in\, \mathbb{Z}^2$, $U(s_x,s_y)$ is a function that introduces poles in $G(s_x, s_y)$ such that $G(s_x,s_y)$ satisfies the alias-cancellation conditions (\ref{eq:alias_cancle_1})-(\ref{eq:alias_cancle_2}), and $V(s_x,s_y)$ is a polynomial function that does not have zeros on $\mathcal{S}$. Next, we choose appropriate $Z(s_x,s_y)$ such that the designed kernels have a compact support. To this end, we invoke the 2-D Paley-Wiener theorem \cite{paley_wiener}, which gives the relation between the growth of entire functions of the exponential type (EFET) in the $s$-domain and the support of their time-domain counterparts. Although in higher dimensions, there are several versions of the theorem, we specifically recall the one given by Gel'fand and Shilov \cite[Chapter 4]{gelfand_generalized_functions}, that fits perfectly in the 2-D design framework.

\begin {theorem} 
\label{theo:paley_multi}
A function $g(x,y) \in L^2(\mathbb{R}^2)$ is compactly supported over the domain $\mathcal{T}=\{|x| \leq \tau_x, |y| \leq \tau_y\}$ if and only if its Laplace transform $G(s_x, s_y)$ is an EFET, that is, there exist a constant $C$ such that $|G(s_x,s_y)| < Ce^{\tau_x |s_x|+\tau_y |s_y|}$, and $G(\mathrm{j}\Omega_x, \mathrm{j}\Omega_y) \in L^2(\mathbb{R}^2)$.
\end{theorem}
In the case of 1-D, the design of a specific class of 1-D SMS kernels for a particular choice of $Z(s_x)$, $V(s_x)$, and $U(s_x)$ was demonstrated in \cite{satish_paley_wiener}. In the 2-D case, we consider two particular choices of $Z(s_x,s_y)$, one each for separable and nonseparable cases, that result in compactly supported kernels. 

\subsection{Separable Kernels}
The design of separable 2-D sampling kernels is a straightforward extension of the 1-D result and is summarized in the following proposition.    
\begin{prop}
\label{prop:separable}
Let $Z(s_x,s_y) = \sinh^{r_1}\left(\frac{\pi s_x}{\Omega_{0x}} \right) \sinh^{r_2}\left(\frac{\pi s_y}{\Omega_{0y}}\right)$ and $U(s_x,s_y) = \hspace{-1.0mm} \prod \limits_{\substack{k_1\in\mathcal{K}_1 \\ k_2\in\mathcal{K}_2}} \hspace{-1.0mm} (s_x-\mathrm{j}k_1\Omega_{0x})^{r_1}(s_y-\mathrm{j}k_2\Omega_{0y})^{r_2}$. Then $\displaystyle G(s_x,s_y)$ satisfies the alias-cancellation conditions \eqref{eq:alias_cancle_1} and \eqref{eq:alias_cancle_2}, and $g(x,y)$ is compactly supported to $\Big[-\frac{r_1\,T_{0x}}{2}, \frac{r_1\,T_{0x}}{2}\Big] \times \left[-\frac{r_2\,T_{0y}}{2}, \frac{r_2\,T_{0y}}{2}\right]$, where $T_{0x}=\frac{2\pi}{\Omega_{0x}}$ and $T_{0y}=\frac{2\pi}{\Omega_{0y}}$. 
\end{prop}
Using tools such as the partial fraction decomposition and the binomial theorem, it can be shown that the impulse response of the kernel proposed in Proposition \ref{prop:separable} is a sum of modulated $r^{\text{th}}$-order separable polynomial B-splines (denoted by $\beta^r(\cdot)$). The proof of the proposition is given in Appendix \ref{appendix:prop_separable}.

For the choice of 
\begin{align}
    V(\mathrm{j}\Omega_{0x},\mathrm{j}\Omega_{0y})= \sum \limits_{\substack{p_1\in\mathcal{K}_1 \\ p_2\in\mathcal{K}_2}}\hspace{-1.0mm}d_{p_1,p_2}\tilde{V}(\mathrm{j}\Omega_{0x},\mathrm{j}\Omega_{0y}),
    \label{eq:separable_V}
\end{align}
where $\tilde{V}(\mathrm{j}\Omega_{0x},\mathrm{j}\Omega_{0y})\hspace{-0.25mm}=\hspace{-1.5mm}\prod \limits_{\substack{k_1\in\mathcal{K}_1\backslash\{p_1\} \\ k_2\in\mathcal{K}_2\backslash\{p_2\}}} \hspace{-2.0mm} (\mathrm{j}\Omega_x-\mathrm{j}k_1\Omega_{0x})^{r_1} (\mathrm{j}\Omega_y-\mathrm{j}k_2\Omega_{0y})^{r_2}$ and with appropriate constants $d_{p_1,p_2}$, we arrive at the special class of compactly supported 2-D SMS kernels proposed in  \cite{anindita_2D_fri}, whose frequency and impulse responses are
\begin{align}
G_{S}(\mathrm{j}\Omega_x, \mathrm{j}\Omega_y) =\hspace{-0.9mm}\sum \limits_{\substack{k_1\in\mathcal{K}_1 \\ k_2\in\mathcal{K}_2}} \hspace{-1.0mm}\text{sinc}^{r_1}\hspace{-1.0mm}\left(\hspace{-0.7mm}\frac{\Omega_x \hspace{-0.5mm}- \hspace{-0.5mm}k_1\,\Omega_{0x}}{\Omega_{0x}}\hspace{-0.7mm}\right)\hspace{-0.7mm}\text{sinc}^{r_2}\hspace{-1.0mm}\left(\hspace{-0.7mm}\frac{\Omega_y \hspace{-0.5mm}-\hspace{-0.5mm} k_2\,\Omega_{0y}}{\Omega_{0y}}\hspace{-0.7mm}\right), \quad \text{and}
 \label{eq:2-D_SMS_freq}
 \end{align}
\begin{align}
g_S(x,y)= \beta^{(r_1-1)}\hspace{-0.7mm}\left(\hspace{-0.7mm}\frac{x}{T_{0x}}\hspace{-0.7mm}\right) \hspace{-0.3mm} \beta^{(r_2-1)}\hspace{-0.7mm}\left(\hspace{-0.7mm}\frac{y}{T_{0y}}\hspace{-0.7mm}\right)
 \hspace{-1.5mm} \sum \limits_{\substack{k_1\in\mathcal{K}_1 \\ k_2\in\mathcal{K}_2}}\hspace{-1.5mm}e^{\mathrm{j}(k_1\Omega_{0x} x+ k_2\Omega_{0y} y)},
 \label{eq:2-D_SMS}
 \end{align}
respectively (Appendix \ref{appendix:kernel_separable_expression}). One could chose different values for $r_1$ and $r_2$ in $g_S(x,y)$ that would result in kernels with different shapes and spatial-domain supports. Figure~\ref{fig:time_dom_separable} shows impulse responses of two such 2-D SMS kernels.\\
\subsubsection{Polynomial and Exponential Reproducing Kernels}
Since the kernel $g_S(x,y)$ in (\ref{eq:2-D_SMS}) is separable, and it was shown in \cite{satish_paley_wiener} that the 1-D SMS kernels satisfy the generalized Strang-Fix conditions \cite{uriguen2013fri}, it is readily seen that the kernel $g_S(x,y)$ could be used to generate a particular class of polynomials/exponentials. Specifically, for separable 2-D SMS kernels, there exist constants $\{c_{k_1,k_2,n_1,n_2}^{i,j}\}$ such that 
$\displaystyle\sum_{n_1,n_2\in \mathbb{Z}}c_{k_1,k_2,n_1,n_2}^{i,j}\,g_S \left( \frac{x-n_1T_{sx}}{T_{sx}}, \frac{y-n_2T_{sy}}{T_{sy}} \right) =$ 
$~\left(\displaystyle\frac{x^i y^j}{T_{sx}^i T_{sy}^j}\right)\, e^{\mathrm{j}k_1\Omega_{0x} x+\mathrm{j}k_2\Omega_{0y} y}$,
for $i \in \llbracket 0, r_1 \rrbracket$ and $j \in \llbracket 0, r_2 \rrbracket$, where $\llbracket a, b \rrbracket$ denotes the set of contiguous integers from $a$ to $b$, both included.

The support of the kernel $g_S(x,y)$ depends on $r_1$ and $r_2$, and is independent of $|\mathcal{K}_1|$ and $|\mathcal{K}_2|$. This is an attractive feature of the 2-D SMS kernels in the sense that they can reproduce exponentials $\{e^{\mathrm{j}k_1\Omega_{0x} x+\mathrm{j}k_2\Omega_{0y} y}\}_{k_1\in \mathcal{K}_1, k_2 \in \mathcal{K}_2}$, wherein the support of the kernel $g_S(x,y)$ is independent of the order.
\begin{figure}[t]
\begin{minipage}[b]{.48\linewidth}
  \centering
  \centerline{\includegraphics[width=0.65\linewidth,trim={2.4cm 0 1cm 0},clip]{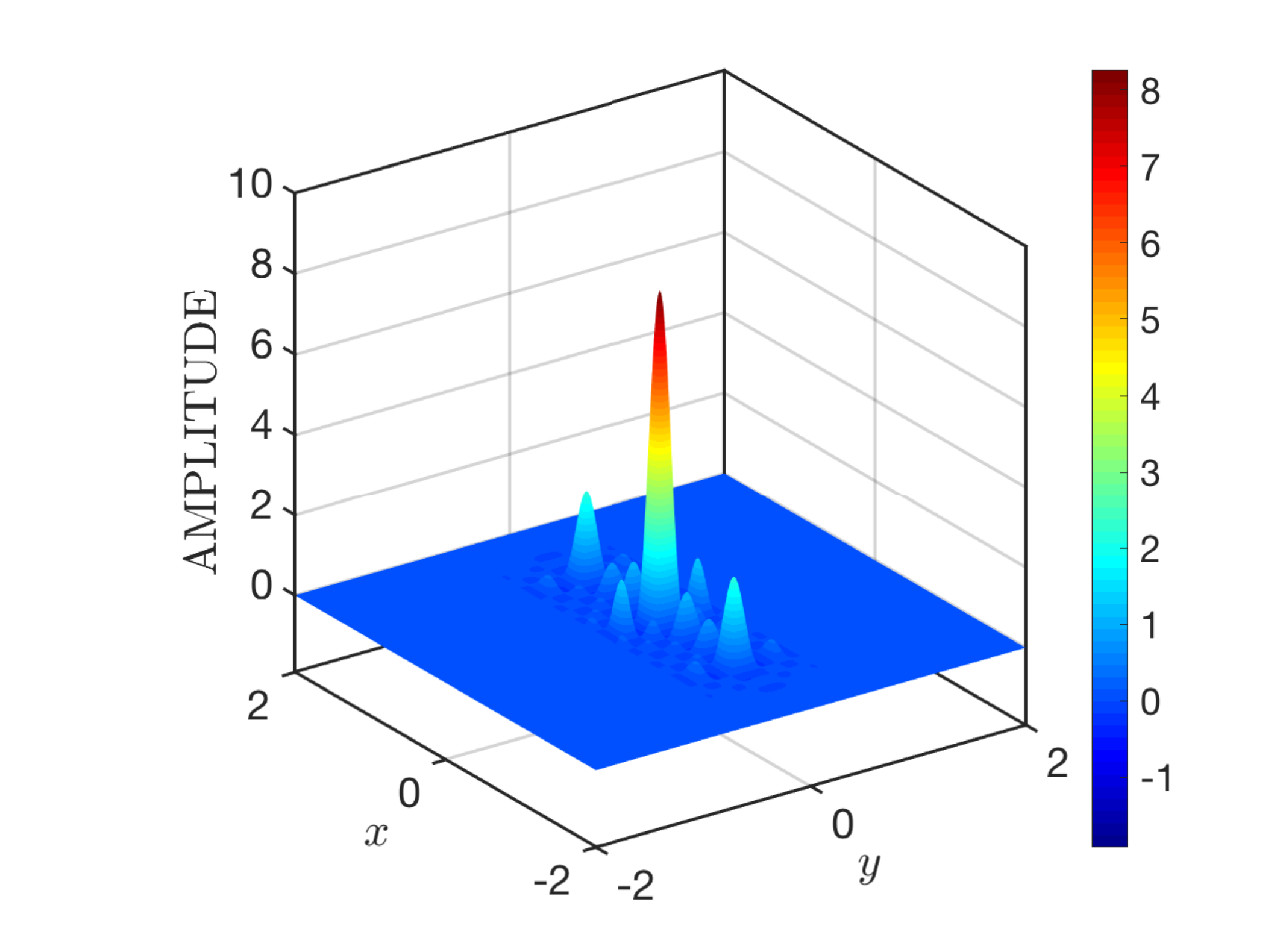}}
  \centerline{(a)}\medskip
\end{minipage}
\hfill
\begin{minipage}[b]{0.48\linewidth}
  \centering
  \centerline{\includegraphics[width=0.65\linewidth,trim={2.4cm 0 1cm 0},clip]{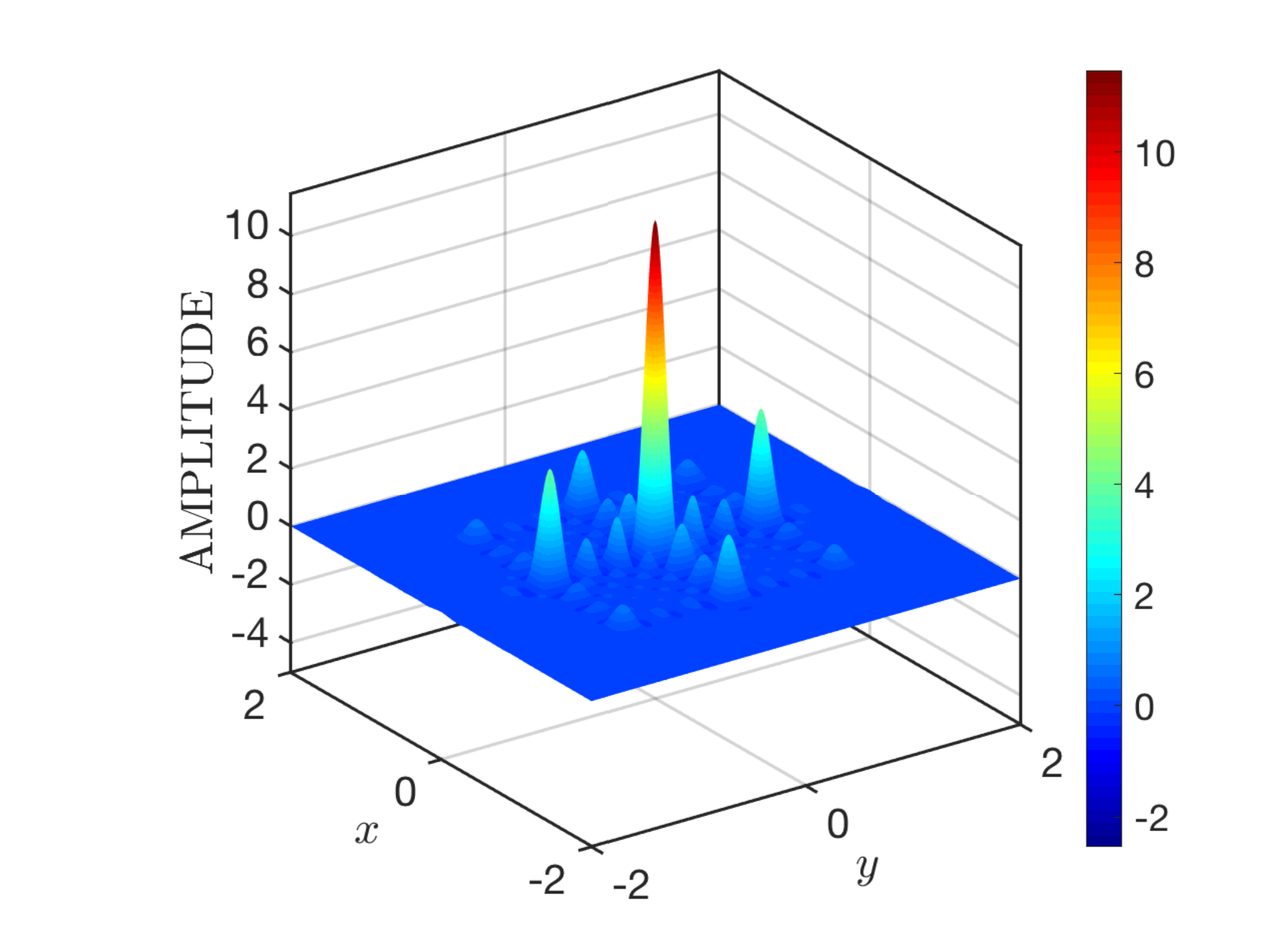}}
  \centerline{(b)}\medskip
\end{minipage}
\hfill
\caption{[Color online] Impulse response of separable 2-D SMS kernel:  $\mathcal{K}_1=\mathcal{K}_2 = \llbracket -L,L\rrbracket$ for $L = 2$, $\Omega_{0x}=\Omega_{0y} = 2\pi$, and (a) $r_1 = 4, r_2 = 1$ and (b) $r_1 = 3, r_2 = 5$.}
\label{fig:time_dom_separable}
\end{figure}
\subsection{Nonseparable Kernels}
We consider a nonseparable $Z(s_x,s_y)$ and a suitable $U(s_x,s_y)$  that results in a nonseparable $G(s_x,s_y)$ and consequently, a compactly supported $g(x,y)$. The following proposition summarizes the result.
\begin{prop}
\label{prop:non_separable}
Let $\displaystyle Z(s_x,s_y)=\sinh\left(\frac{\pi s_x}{\Omega_{0x}} + \frac{\pi s_y}{\Omega_{0y}}\right)\sinh\left(\frac{\pi s_y}{\Omega_{0y}} - \frac{\pi s_x}{\Omega_{0x}}\right)\, \text{and}\\ \,$ $ U(s_x,s_y) \hspace{-0.25mm}= \hspace{-1.5mm}\prod \limits_{\substack{k_1\in\mathcal{K}_1 \\ k_2\in\mathcal{K}_2}}\hspace{-1.5mm}\left(\hspace{-0.7mm}\frac{s_x}{\Omega_{0x}} \hspace{-0.3mm}+\hspace{-0.3mm} \frac{s_y}{\Omega_{0y}} \hspace{-0.3mm}-\hspace{-0.3mm} \mathrm{j}k_1 \hspace{-0.3mm}-\hspace{-0.3mm} \mathrm{j}k_2\hspace{-0.7mm} \right)  \left(\hspace{-0.7mm}\frac{s_y}{\Omega_{0y}} \hspace{-0.3mm}-\hspace{-0.3mm} \frac{s_x}{\Omega_{0x}} \hspace{-0.3mm}-\hspace{-0.3mm} \mathrm{j}k_2 \hspace{-0.3mm}+\hspace{-0.3mm} \mathrm{j}k_1 \hspace{-0.7mm}\right)$. Then $\displaystyle G(s_x,s_y)$ satisfies the alias-cancellation conditions \eqref{eq:alias_cancle_1} and \eqref{eq:alias_cancle_2}, and $g(x,y)$ is supported over $\left[-{T_{0x}}, {T_{0x}}\right] \times \left[-{T_{0y}}, {T_{0y}}\right]$, where $T_{0x}=\frac{2\pi}{\Omega_{0x}}$ and $T_{0y}=\frac{2\pi}{\Omega_{0y}}$. 
\end{prop}
The proof involves two steps: showing that (i) $G(s_x,s_y)$ is an EFET; and (ii) $G(\mathrm{j}\Omega_x, \mathrm{j}\Omega_y)\in L^2(\mathbb{R}^2)$, and is provided in Appendix \ref{appendix:prop_nonseparable}. Substituting the functions $Z(s_x,s_y)$ and $U(s_x,s_y)$ of Proposition \ref{prop:non_separable} in (\ref{eq:gen_samp_kernel}), and using the partial fraction decomposition of $\frac{V(s_x,s_y)}{U(s_x,s_y)}$ and rotation property of the 2-D CTFT, it can be shown that (see Appendix \ref{appendix:kernel_nonseparable_expression}) the frequency response of the sampling kernel in Proposition \ref{prop:non_separable} is given by 
\begin{align}
G_{NS}(\mathrm{j}\Omega_x, \mathrm{j}\Omega_y)= \,\, \pi^2 \sum_{k_1\in\mathcal{K}_1}\sum_{k_2\in\mathcal{K}_2}q_{k_1,k_2} \text{sinc}\left(\frac{\Omega_x - k_1\,\Omega_{0x}}{\Omega_{0x}}+\frac{\Omega_y - k_2\,\Omega_{0y}}{\Omega_{0y}}\right)
\text{sinc}\left(\frac{\Omega_y - k_2\,\Omega_{0y}}{\Omega_{0y}}\,-\, \frac{\Omega_x - k_1\,\Omega_{0x}}{\Omega_{0x}}\right), 
\label{eq:G_omega_non_sepa}
\end{align}
where $ q_{k_1,k_2}$ are the partial fraction decomposition coefficients. The corresponding spatial-domain kernel is 
\begin{align}
g_{NS}(x,y)=\frac{\Omega_{0x}\Omega_{0y}}{8}\,\,\text{rect}\left( \frac{\Omega_{0x}x + \Omega_{0y}y}{4\pi}\right)\text{rect}\left( \frac{\Omega_{0y}y - \Omega_{0x}x}{4\pi}\right)  \sum_{k_1\in\mathcal{K}_1}\sum_{k_2\in\mathcal{K}_2}q_{k_1,k_2}\, e^{\mathrm{j}(k_1\Omega_{0x}x + k_2\Omega_{0y}y)},
\label{eq:s_non_sepa}
\end{align}
where $\text{rect}(\cdot)$ is defined as  $\text{rect}(\frac{x}{T_x}) = 1$, if $|x| \leq T_x/2$, and $0$ otherwise. The sampling kernel derived in (\ref{eq:s_non_sepa}) is for a generic polynomial function $V(s_x,s_y)$, except that it does not have zeros on $\mathcal{S}$. For the particular choice \\
$V(\mathrm{j}\Omega_x, \mathrm{j}\Omega_y) =\hspace{-1.5mm}\sum \limits_{\substack{p_1\in\mathcal{K}_1 \\ p_2\in\mathcal{K}_2}}\hspace{-1mm} d_{p_1,p_2} 
\hspace{-1mm} \prod \limits_{\substack{k_1\in\mathcal{K}_1\backslash\{p_1\} \\ k_2\in\mathcal{K}_2\backslash\{p_2\}}}
\hspace{-1mm}\left(\frac{\mathrm{j}\,\Omega_x}{\Omega_{0x}} + \frac{\mathrm{j}\,\Omega_y}{\Omega_{0y}} - \mathrm{j}k_1 - \mathrm{j}k_2 \right) ~ \left(\frac{\mathrm{j}\,\Omega_y}{\Omega_{0y}} - \frac{\mathrm{j}\,\Omega_x}{\Omega_{0x}} - \mathrm{j}k_2 + \mathrm{j}k_1 \right)$,
we get $q_{k_1,k_2} = d_{k_1,k_2}$. \\
\subsubsection{Discussion}
Even though the nonseparable kernel $g_{NS}(x,y)$ in (\ref{eq:s_non_sepa}) appears to be a rotated version of the separable kernel $g_S(x,y)$ in (\ref{eq:2-D_SMS}) with $r_1=r_2=1$, a closer observation reveals something more. If we rotate the kernel $g_S(x,y)$ or equivalently $G_S(\mathrm{j}\Omega_x, \mathrm{j}\Omega_y)$, then the zeros of the kernel will also shift in the 2-D plane. Hence, the alias cancellation conditions specified in (\ref{eq:alias_cancle_1}) and (\ref{eq:alias_cancle_2}) are no more satisfied on the 2-D rectangular grid, but are valid on the rotated 2-D grid. This means, to counter the effect of rotation, the sampling mechanism and the reconstruction techniques have to be suitably modified. On the other hand, in the case of the proposed nonseparable kernel $G_{NS}(\mathrm{j}\Omega_x, \mathrm{j}\Omega_y)$ in (\ref{eq:G_omega_non_sepa}), the alias cancellation conditions are met on the 2-D grid, and the usual sampling and reconstruction techniques that are applicable in the case of separable kernels could be deployed. 

Unlike the separable kernels, in the case of a nonseparable kernel $g_{NS}(x,y)$, we have considered only polynomial B-splines of zeroth order. Developing nonseparable kernels that are a sum of modulated-splines of higher orders needs more investigation and makes an interesting case for future work. Design of such kernels might result in their impulse responses being non-isotropic, and could be employed to approximate a more wider class of point-spread functions in the imaging modalities such as localization microscopy, radio astronomy, etc. Another aspect that is worthy for future investigation is the analysis and design of nonseparable  kernels that reproduce exponentials of the form $e^{\mathrm{j}\alpha xy}$ for some $\alpha$.

\section{Simulation Results}
\begin{figure}[t]
\begin{minipage}[b]{.48\linewidth}
  \centering
  \centerline{\includegraphics[width=0.65\linewidth,trim={2.4cm 0 1cm 0},clip]{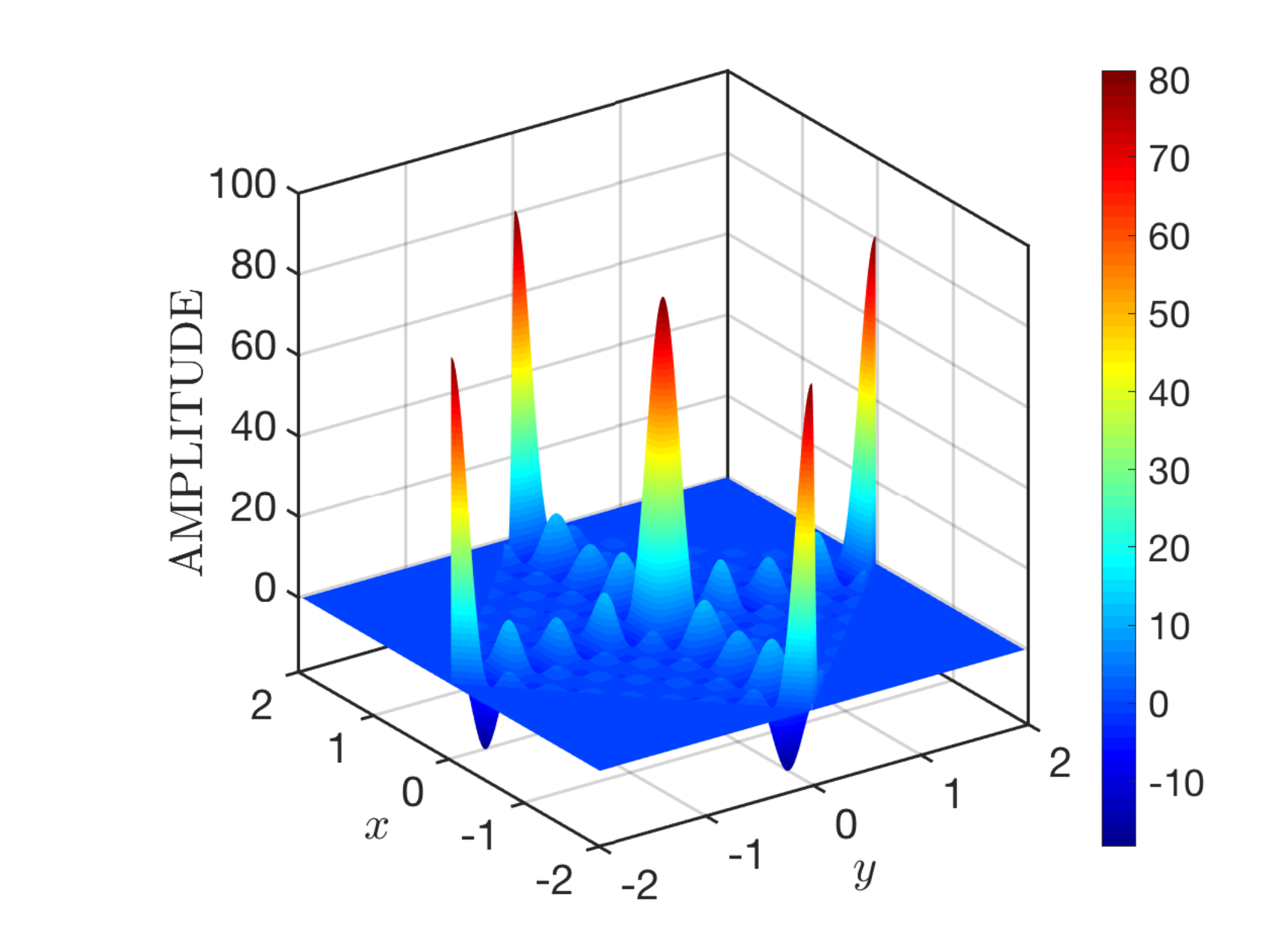}}
  \centerline{(a)}\medskip
\end{minipage}
\hfill
\begin{minipage}[b]{0.48\linewidth}
  \centering
  \centerline{\includegraphics[width=0.65\linewidth,trim={2.4cm 0 1cm 0},clip]{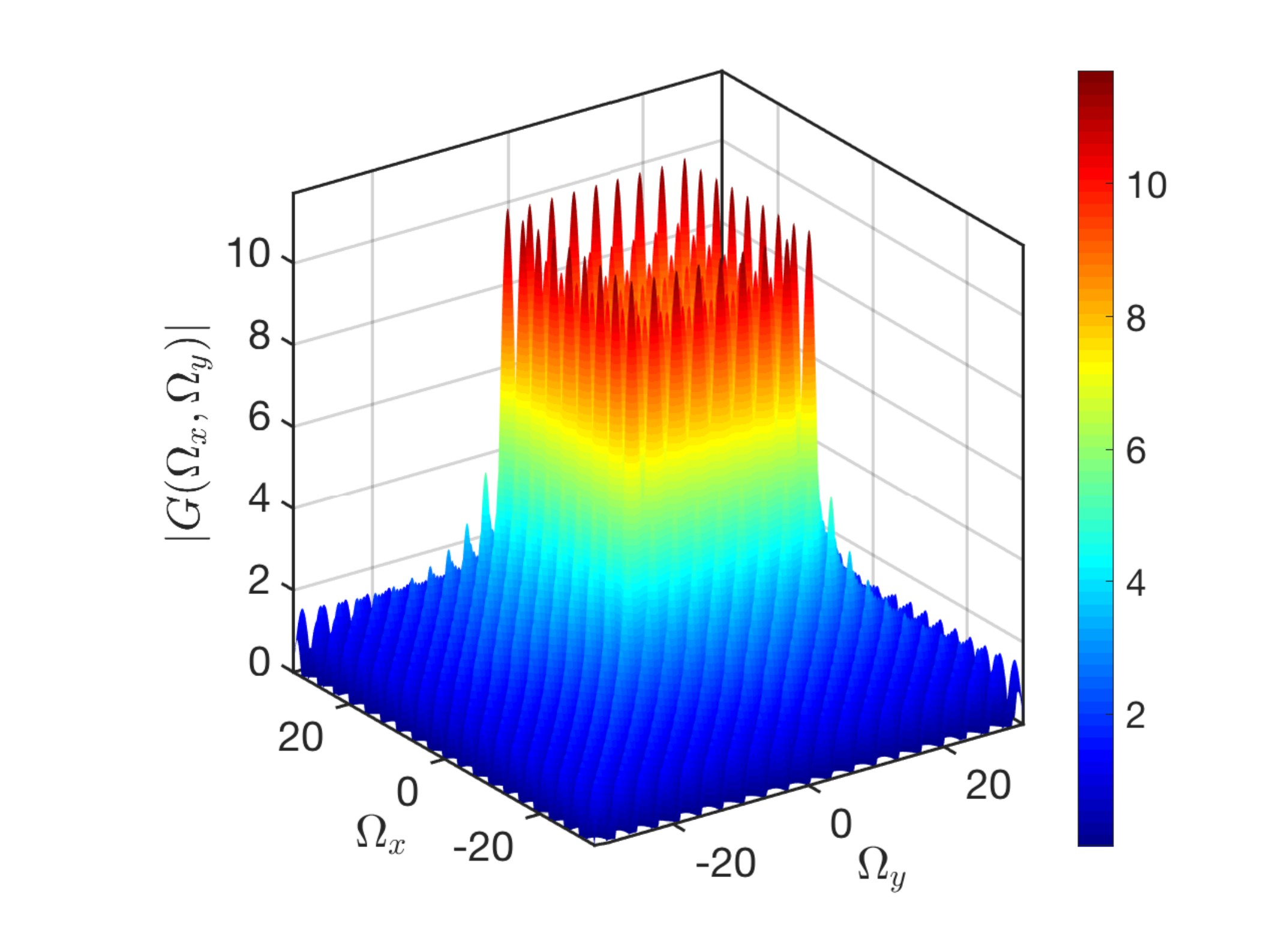}}
  \centerline{(b)}\medskip
\end{minipage}
\begin{minipage}[b]{0.48\linewidth}
  \centering
  \centerline{\includegraphics[width=0.65\linewidth,trim={2.4cm 0 1cm 0},clip]{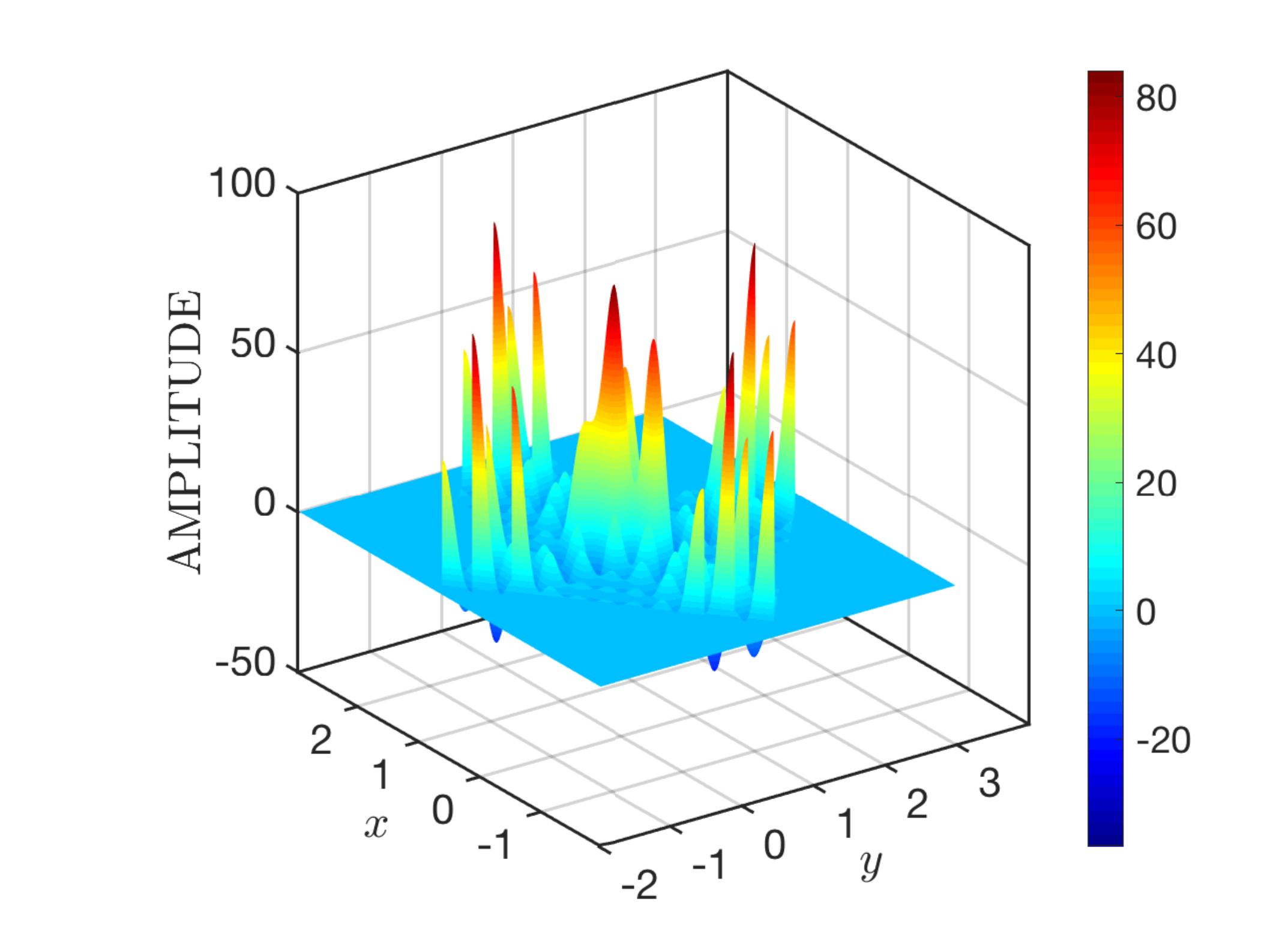}}
  \centerline{(c)}\medskip
\end{minipage}
\hfill
\begin{minipage}[b]{0.48\linewidth}
  \centering
  \centerline{\includegraphics[width=0.65\linewidth,trim={2.4cm 0 1cm 0},clip]{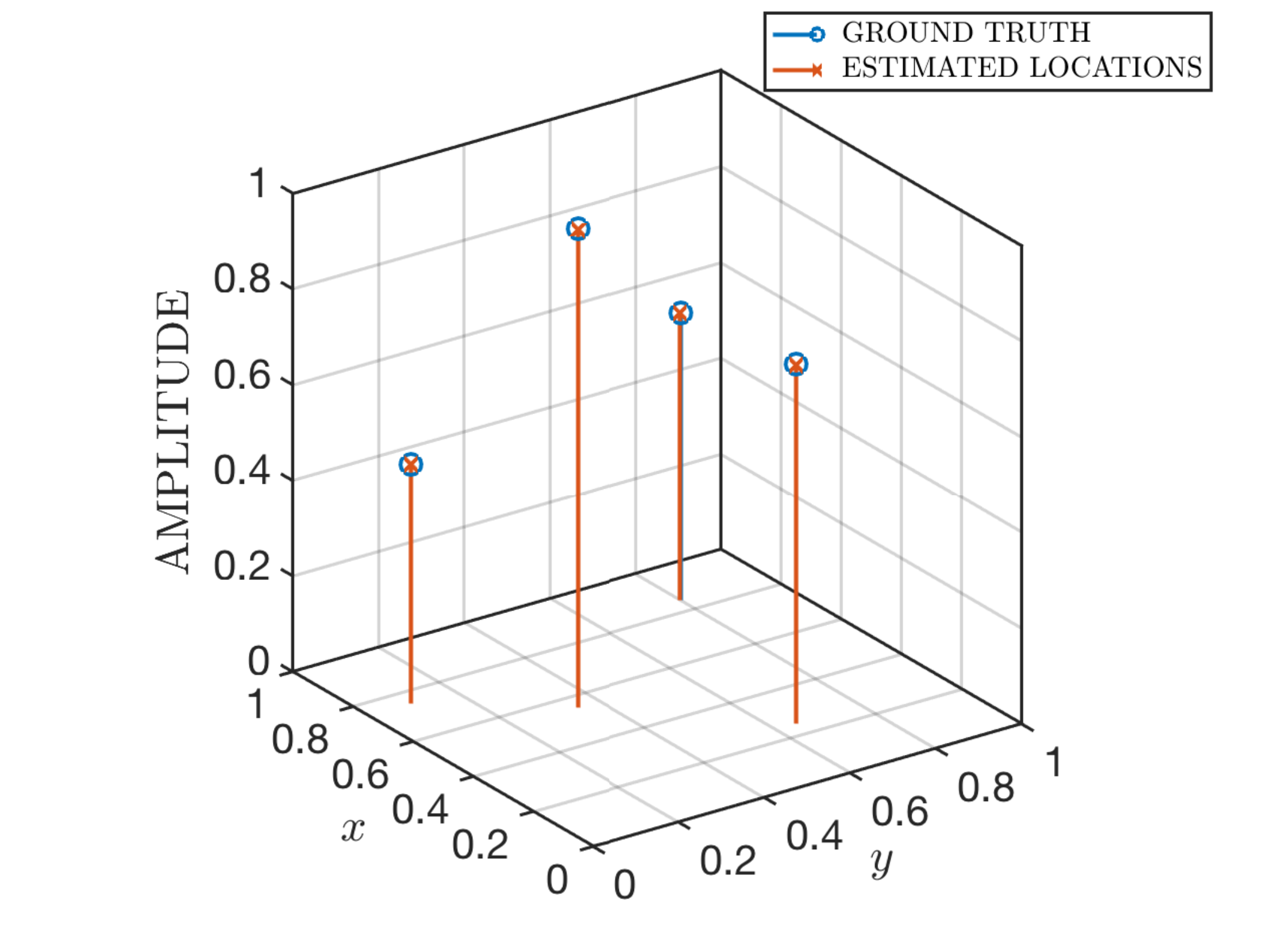}}
  \centerline{(d)}\medskip
\end{minipage}
\hfill
\caption{[Color online] (a) Impulse response; (b) Frequency response of the nonseparable sampling kernel $g_{NS}(x,y)$ in (\ref{eq:s_non_sepa}) with $\mathcal{K}_1 =\mathcal{K}_2 = \llbracket -L,L\rrbracket$ for $L=4$, and $\Omega_{0x} =\Omega_{0y} = \frac{\pi}{0.99}$; (c) Spatial-domain signal obtained by convolving the Diracs with the sampling kernel; and (d) The ground truth and reconstructed Dirac locations.}
\label{fig:dirac_sim}
\end{figure}

\begin{figure}[t]
\begin{minipage}[b]{.48\linewidth}
  \centering
  \centerline{\includegraphics[width=0.65\linewidth,trim={2.4cm 0 1cm 0},clip]{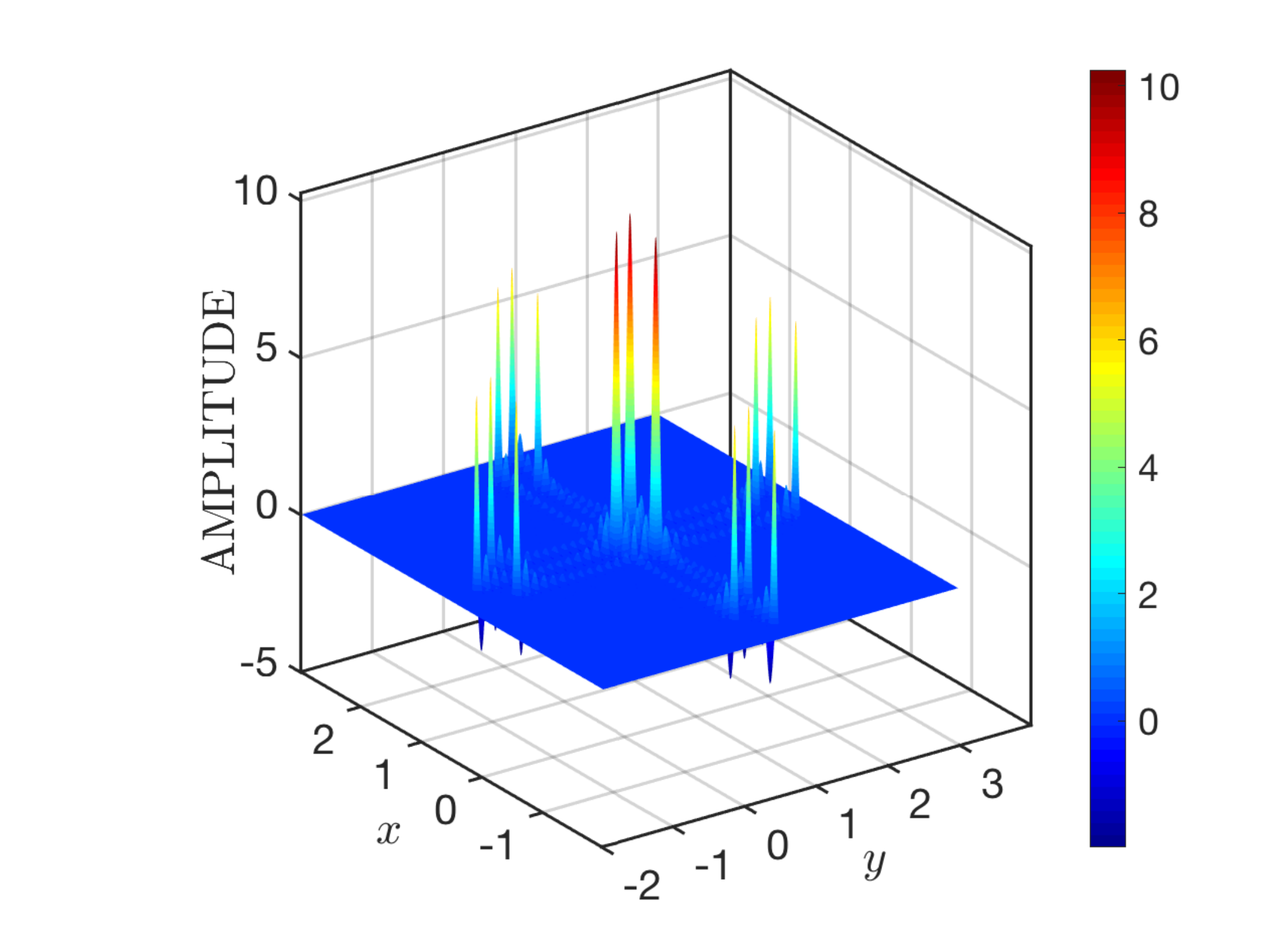}}
 \vspace{-2mm}
  \centerline{(a)}\medskip
\end{minipage}
\hfill
\begin{minipage}[b]{0.48\linewidth}
  \centering
  \centerline{\includegraphics[width=0.65\linewidth,trim={1.4cm 0 1cm 0},clip]{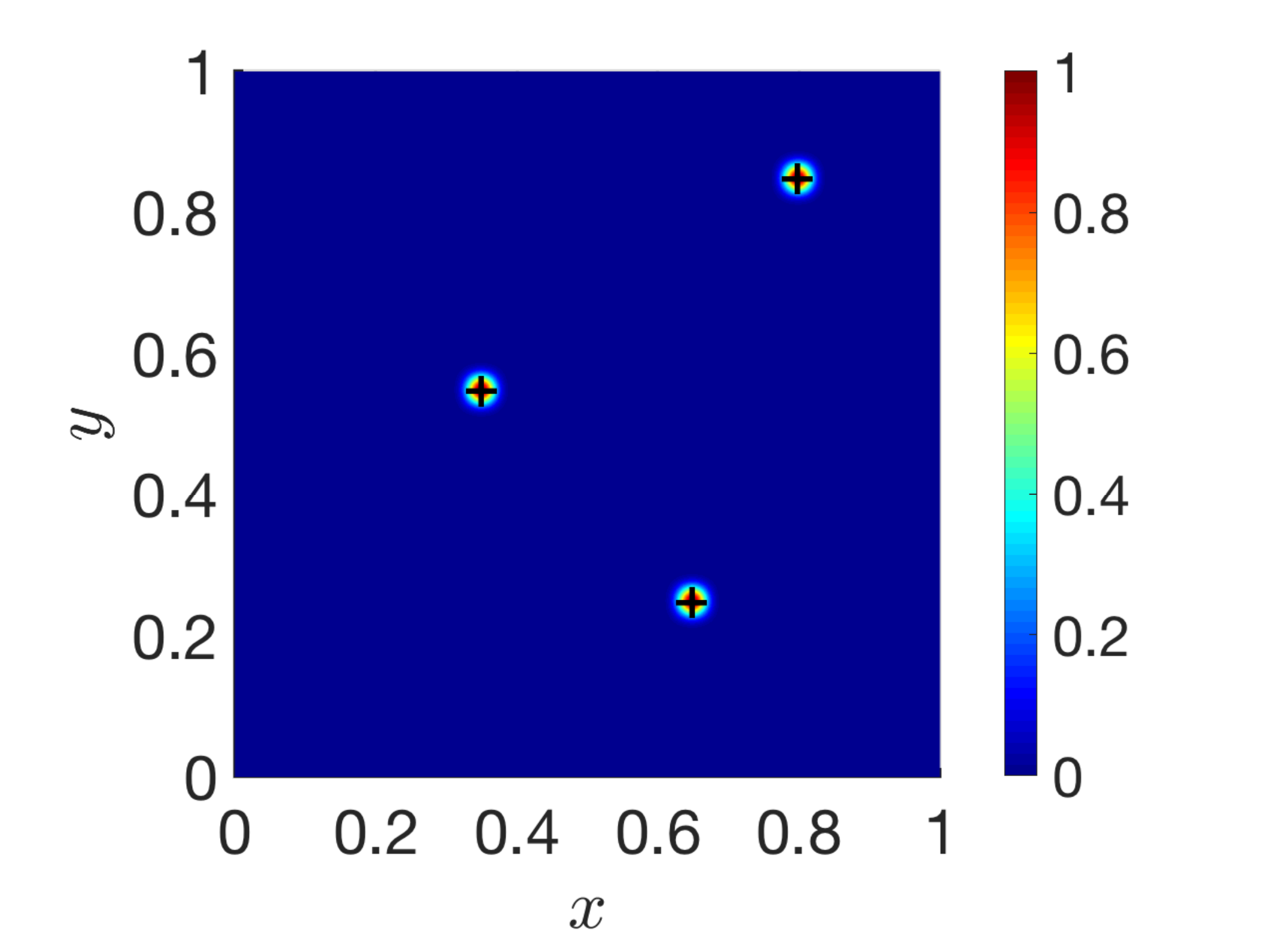}}
  \vspace{-2mm}
  \centerline{(b)}\medskip
\end{minipage}
\hfill
\caption{[Color online] Localization of Gaussian blobs: (a) Result of convolution of three ($L=3$) truncated Gaussian functions with the nonseparable sampling kernel with $\mathcal{K}_1 = \mathcal{K}_2 = \llbracket -5L,5L\rrbracket$ and $\Omega_{0x} = \Omega_{0y} = \frac{\pi}{0.99}$; and  (b) Original 2-D FRI signal with the estimated locations (marked `+'). The SNR of the measurements used for reconstruction was $15$ dB and the MSE in the estimation of locations (over 50 realizations) was computed to be $-55$ dB.}
\label{fig:spline_sim}
\end{figure}

We conducted two experiments to demonstrate the applicability of the proposed non-separable sampling kernel. In the first experiment, we consider a 2-D signal having four Diracs $(L = 4)$ with parameters $\{\gamma_\ell,x_\ell,\, y_\ell\}_{\ell = 1}^{4}$ that are selected uniformly at random over $(0\,,\,1)$. A compactly supported nonseparable kernel as in (\ref{eq:s_non_sepa}) with the support restricted to $\left[-{T_{0x}}, {T_{0x}}\right] \times \left[-{T_{0y}}, {T_{0y}}\right]$ is simulated with the following parameters:  $\mathcal{K}_1 = \mathcal{K}_2 =\llbracket -4,4\rrbracket$,  $\Omega_{0x} = \Omega_{0y} =\frac{\pi}{0.99}$. The spatial domain samples are acquired at the critical sampling rate of  $\Omega_{sx} = (2L + 1)\Omega_{0x}$ and $\Omega_{sy} = (2L + 1)\Omega_{0y}$. The parameters were estimated using the algebraically coupled matrix pencil reconstruction method \cite{2DHarmonicRetrieval_vanpoucke1994}. The input signal along with the estimated Dirac locations, the impulse and frequency response of the sampling kernel, and convolution output of the signal and the sampling kernel are shown in Fig. \ref{fig:dirac_sim}. The mean-square error in the estimation of Dirac locations was computed to be $-280$ dB implying perfect reconstruction up to machine precision. 

In the second experiment, we consider a 2-D FRI signal consisting of  three truncated Gaussian functions located at $\{x_\ell,\, y_\ell\}_{\ell = 1}^{3}$ that are selected uniformly at random over $(0\,,\,1)$. Zero-mean,  additive white Gaussian noise is added to the spatial-domain samples such that the signal-to-noise ratio (SNR) of the resulting samples is $15$ dB. We oversample the signal with the sampling rates $\Omega_{sx} = (10L + 1)\Omega_{0x}$ and $\Omega_{sy} = (10L + 1)\Omega_{0y}$. Figure \ref{fig:spline_sim} shows the ground truth signal with the accurately estimated locations of the Gaussian blobs. A comprehensive assessment of noise robustness of the kernels vis-\`a-vis the Cram\'er-Rao bounds will be addressed in a future work.

\section{Conclusions}
In this paper, we proposed a generalized framework for designing compactly supported sampling kernels for 2-D FRI signals. The first key idea in this generalization was to design the frequency response of the kernel, which satisfies a set of alias cancellation conditions; and second, to characterize admissible kernels with a compact spatial support by invoking the 2-D Paley-Wiener theorem.  The proposed framework allows for the design of both separable and nonseparable 2-D sampling kernels. As a particular case, we showed that a special case of the separable sampling kernel is the class of 2-D SMS kernels,  which has the attractive feature of reproducing a certain class of exponentials and  the support of the kernel is independent of the order. We also demonstrated the design of a nonseparable kernel
and validated it by performing simulations to extract the exact locations of Diracs in the 2-D plane. The design of higher-order nonseparable kernels and analysis of their exponential reproducing properties are some interesting aspects for further study.


%


\appendices
\section{Proof of Proposition \ref{prop:separable}}
\label{appendix:prop_separable}

The Proof of Proposition \ref{prop:separable} involves two steps: (i) showing $G_S(s_x,s_y)$ is an EFET; and (ii) showing $G_S(\mathrm{j}\Omega_x, \mathrm{j}\Omega_y) \in L^2(\mathbb{R})$.
\subsection{Proof that $G_S(s_x,s_y)$ is an EFET}
\label{appendix:G_efet_separable}
\begin{proof}[Proof]

For the separable case, we have
\begin{equation}
    Z_{r_1,r_2}(s_x,s_y) = \sinh^{r_1}\left(\frac{\pi s_x}{\Omega_{0x}} \right) \sinh^{r_2}\left(\frac{\pi s_y}{\Omega_{0y}}\right). \nonumber
\end{equation}

We know that

\begin{align}
    \left|\sinh\left(\frac{\pi s_x}{\Omega_{0x}} \right) \sinh\left(\frac{\pi s_y}{\Omega_{0y}}\right)\right| \leq& \left|\sinh\left(\frac{\pi s_x}{\Omega_{0x}} \right)\right| \left|\sinh\left(\frac{\pi s_y}{\Omega_{0y}}\right)\right|, \nonumber \\ 
    =& \left| \frac{e^{\frac{\pi s_x}{\Omega_{0x}}} - e^{ - \frac{\pi s_x}{\Omega_{0x}}} }{2} \right| \left|\frac{e^{\frac{\pi s_y}{\Omega_{0y}}} - e^{ - \frac{\pi s_y}{\Omega_{0y}}} }{2} \right|, \nonumber \\ 
    \leq& \frac{1}{4} \left( \left| e^{\frac{\pi s_x}{\Omega_{0x}}}\right| + \left|e^{ - \frac{\pi s_x}{\Omega_{0x}}} \right| \right) \left(\left|e^{\frac{\pi s_y}{\Omega_{0y}}} \right| + \left| e^{-\frac{\pi s_y}{\Omega_{0y}}} \right|\right), \nonumber \\
    \leq& \frac{1}{4} \left(  e^{\left|\frac{\pi s_x}{\Omega_{0x}}\right|} + e^{\left|\frac{\pi s_x}{\Omega_{0x}}\right|}  \right) \left(e^{\left|\frac{\pi s_y}{\Omega_{0y}}\right|}  +  e^{\left|\frac{\pi s_y}{\Omega_{0y}}\right|} \right), \nonumber \\
    =& \frac{1}{4} \left( 2\, e^{\pi\left|\frac{s_x}{\Omega_{0x}}\right|}\right) \left( 2\, e^{\pi\left|\frac{ s_y}{\Omega_{0y}}\right|}\right), \nonumber \\
    =& \, e^{\left(\frac{\pi}{\Omega_{0x}}\left|s_x \right|  + \frac{\pi}{\Omega_{0y}}\left|s_y\right|\right)} . \nonumber
\end{align}

Consequently, 
\begin{equation}
    \left|\sinh^{r_1}\left(\frac{\pi s_x}{\Omega_{0x}} \right) \sinh^{r_2}\left(\frac{\pi s_y}{\Omega_{0y}}\right) \right| \leq \,  e^{\left(\frac{r_1\pi}{\Omega_{0x}} \left|s_x \right|  + \frac{r_2\pi}{\Omega_{0y}}\left|s_y\right|\right)}. \nonumber
\end{equation}

Since $Z(\mathrm{j}\Omega_x, \mathrm{j}\Omega_y)$  is a sum of entire functions, it is an EFET, and consequently $G_S(\mathrm{j}\Omega_x, \mathrm{j}\Omega_y)$ is also an EFET (as shown in \cite{satish_paley_wiener}). The support of $g_S(x,y)$ is given by $\left[-\frac{r_1\pi}{\Omega_{0x}}, \frac{r_1\pi}{\Omega_{0x}}\right] \times \left[-\frac{r_2\pi}{\Omega_{0y}}, \frac{r_2\pi}{\Omega_{0y}}\right]$ i.e., $\Big[-\frac{r_1\,T_{0x}}{2}, \frac{r_1\,T_{0x}}{2}\Big] \times \left[-\frac{r_2\,T_{0y}}{2}, \frac{r_2\,T_{0y}}{2}\right]$.

\end{proof}

\subsection{Proof that $G_S(\mathrm{j}\Omega_x, \mathrm{j}\Omega_y) \in L^2(\mathbb{R})$}
\label{appendix:G_energy_separable}
\begin{proof}[Proof]


Using the hyperbolic sine identity: $\sinh(\mathrm{j}\theta)=\mathrm{j}\sin(\theta)$, $\forall \,\theta\in \mathbb{R}$, $G_S(\mathrm{j}\Omega_{0x}, \mathrm{j}\Omega_{0y})$ can be expressed as
\begin{equation}
\centering
G_S(\mathrm{j}\Omega_{0x},\mathrm{j}\Omega_{0y}) = V(\mathrm{j}\Omega_{0x},\mathrm{j}\Omega_{0y})\,\mathrm{j}^{(r_1 (1 - |\mathcal{K}_1|) + r_2(1 - |\mathcal{K}_2|))}\frac{\sin^{r_1}\left(\frac{\pi \Omega_x}{\Omega_{0x}} \right) \sin^{r_2}\left(\frac{\pi \Omega_y}{\Omega_{0y}}\right)}{\prod \limits_{\substack{k_1\in\mathcal{K}_1 \\ k_2\in\mathcal{K}_2}} (\Omega_{x}-k_1\Omega_{0x})^{r_1}(\Omega_{y}-k_2\Omega_{0y})^{r_2}}. \nonumber
\end{equation}
Using partial fraction decomposition, $G_S(\mathrm{j}\Omega_{0x}, \mathrm{j}\Omega_{0y})$ takes the form
\begin{equation}
\centering
G_S(\mathrm{j}\Omega_{0x},\mathrm{j}\Omega_{0y})=\sum_{k_1\in\mathcal{K}_1}\sum_{k_2\in\mathcal{K}_2}\sum_{m_1=1}^{r_1}\sum_{m_2=1}^{r_2}\frac{q_{k_1,k_2,m_1,m_2}}{(\Omega_{x}-k_1\Omega_{0x})^{m_1}(\Omega_{y}-k_1\Omega_{0y})^{m_2}}\sin^{r_1}\left(\frac{\pi \Omega_x}{\Omega_{0x}} \right) \sin^{r_2}\left(\frac{\pi \Omega_y}{\Omega_{0y}}\right),
\label{gs_partial}
\end{equation}
where $q_{k_1,k_2,m_1,m_2}$ are the coefficients of the partial fraction expansion and are given by
\begin{equation}
\centering
q_{k_1,k_2,m_1,m_2}=\frac{\mathrm{j}^{(r_1 (1 - |\mathcal{K}_1|) + r_2(1 - |\mathcal{K}_2|))}}{(r_1-m_1)!(r_2-m_2)!}\frac{\text{d}^{(r_1-m_1)}}  {\text{d}\Omega_{x}^{(r_1-m_1)}} \frac{\text{d}^{(r_2-m_2)}}  {\text{d}\Omega_{y}^{(r_2-m_2)}} \left. \left( \frac{V(\mathrm{j}\Omega_{0x},\mathrm{j}\Omega_{0y})}{\prod \limits_{\substack{p_1\in\mathcal{K}_1\backslash\{k_1\} \\ p_2\in\mathcal{K}_2\backslash\{k_2\}}} (\Omega_{x}-p_1\Omega_{0x})^{r_1}(\Omega_{y}-p_1\Omega_{0y})^{r_2}} \right ) \right|_{\substack{\Omega_x=k_1\Omega_{0x} \\ \Omega_y=k_2\Omega_{0y}}},\nonumber
\end{equation}
for $m_1= 1, 2, \cdots, r_1$, and $m_2= 1, 2, \cdots, r_2$.

Using the Cauchy-Schwarts inequality, we have
\begin{align}
    \int\limits_{-\infty}^{\infty}\int\limits_{-\infty}^{\infty}|G_S(\mathrm{j}\Omega_x, \mathrm{j}\Omega_y)|^2 \text{d}\Omega_x\text{d}\Omega_y = & \int\limits_{-\infty}^{\infty}\int\limits_{-\infty}^{\infty} \left|\sum_{k_1\in\mathcal{K}_1}\sum_{k_2\in\mathcal{K}_2}\sum_{m_1=1}^{r_1}\sum_{m_2=1}^{r_2}q_{k_1,k_2,m_1,m_2}\frac{\sin^{r_1}\left(\frac{\pi \Omega_x}{\Omega_{0x}} \right) \sin^{r_2}\left(\frac{\pi \Omega_y}{\Omega_{0y}}\right)}{(\Omega_{x}-k_1\Omega_{0x})^{m_1}(\Omega_{y}-k_1\Omega_{0y})^{m_2}}\right|^2\text{d}\Omega_x\text{d}\Omega_y, \nonumber \\
    \leq& \,\,r_1 r_2 |\mathcal{K}_1||\mathcal{K}_2|\sum_{k_1\in\mathcal{K}_1}\sum_{k_2\in\mathcal{K}_2}\sum_{m_1=1}^{r_1}\sum_{m_2=1}^{r_2}\,|q_{k_1,k_2,m_1,m_2}|^2 \nonumber \\
    & \int\limits_{-\infty}^{\infty}\int\limits_{-\infty}^{\infty}\left|\frac{\sin^{r_1}\left(\frac{\pi \Omega_x}{\Omega_{0x}} \right) \sin^{r_2}\left(\frac{\pi \Omega_y}{\Omega_{0y}}\right)}{(\Omega_{x}-k_1\Omega_{0x})^{m_1}(\Omega_{y}-k_1\Omega_{0y})^{m_2}}\right|^2 \text{d}\Omega_x\text{d}\Omega_y, \nonumber \\
    = & \,\,r_1 r_2 |\mathcal{K}_1||\mathcal{K}_2|\sum_{k_1\in\mathcal{K}_1}\sum_{k_2\in\mathcal{K}_2}\sum_{m_1=1}^{r_1}\sum_{m_2=1}^{r_2}\,|q_{k_1,k_2,m_1,m_2}|^2  \int\limits_{-\infty}^{\infty}\left|\frac{\sin^{r_1}\left(\frac{\pi \Omega_x}{\Omega_{0x}} \right) }{(\Omega_{x}-k_1\Omega_{0x})^{m_1}}\right|^2 \text{d}\Omega_x \nonumber \\
    & \int\limits_{-\infty}^{\infty}\left|\frac{\sin^{r_2}\left(\frac{\pi \Omega_y}{\Omega_{0y}}\right)}{(\Omega_{y}-k_1\Omega_{0y})^{m_2}}\right|^2\text{d}\Omega_y. \nonumber
\end{align}

As shown in \cite{satish_paley_wiener}, $\int\limits_{-\infty}^{\infty}\left|\frac{\sin^{r_1}\left(\frac{\pi \Omega_x}{\Omega_{0x}} \right) }{(\Omega_{x}-k_1\Omega_{0x})^{m_1}}\right|^2 \text{d}\Omega_x < \infty$ and $\int\limits_{-\infty}^{\infty}\left|\frac{\sin^{r_2}\left(\frac{\pi \Omega_y}{\Omega_{0y}}\right)}{(\Omega_{y}-k_1\Omega_{0x})^{m_2}}\right|^2\text{d}\Omega_y < \infty$. Hence,  we have \\ $\int\limits_{-\infty}^{\infty}\int\limits_{-\infty}^{\infty}|G_S(\mathrm{j}\Omega_x, \mathrm{j}\Omega_y)|^2 \text{d}\Omega_x\text{d}\Omega_y < \infty$.

\end{proof}

\subsection{Expression for $g_S(x,y)$}
\label{appendix:kernel_separable_expression}

Consider 
\begin{equation}
\centering
G_S(\mathrm{j}\Omega_{0x},\mathrm{j}\Omega_{0y})=\sum_{k_1\in\mathcal{K}_1}\sum_{k_2\in\mathcal{K}_2}\sum_{m_1=1}^{r_1}\sum_{m_2=1}^{r_2}\frac{q_{k_1,k_2,m_1,m_2}}{(\Omega_{x}-k_1\Omega_{0x})^{m_1}(\Omega_{y}-k_1\Omega_{0y})^{m_2}}\sin^{r_1}\left(\frac{\pi \Omega_x}{\Omega_{0x}} \right) \sin^{r_2}\left(\frac{\pi \Omega_y}{\Omega_{0y}}\right).
\end{equation}
Further simplifying the above equation, we get 
\begin{align}
G_S(\mathrm{j}\Omega_{0x},\mathrm{j}\Omega_{0y})=&\sum_{k_1\in\mathcal{K}_1}\sum_{k_2\in\mathcal{K}_2}\sum_{m_1=1}^{r_1}\sum_{m_2=1}^{r_2} q_{k_1,k_2,m_1,m_2} (\Omega_x-k\Omega_{0x})^{r_1-m_1}(\Omega_y-k_2\Omega_{0y})^{r_2-m_2}\,\, \frac{\sin^{r_1}\left( \frac{\pi \Omega_x}{\Omega_{0x}} \right)}{(\Omega_x-k_1\Omega_{0x})^{r_1}}\frac{\sin^{r_2}\left( \frac{\pi \Omega_y}{\Omega_{0y}} \right)}{(\Omega_y-k\Omega_{0y})^{r_2}}, \nonumber \\
=&\sum_{k_1\in\mathcal{K}_1}\sum_{k_2\in\mathcal{K}_2}\sum_{m_1=1}^{r_1}\sum_{m_2=1}^{r_2} q_{k_1,k_2,m_1,m_2}(\Omega_x-k\Omega_{0x})^{r_1-m_1}(\Omega_y-k_2\Omega_{0y})^{r_2-m_2}(-1)^{r_1k_1 + r_2k_2} \nonumber \\
&\frac{\sin^{r_1}\left( \frac{ \Omega_x-k_1\Omega_{0x}}{\Omega_{0x}}\pi \right)}{(\Omega_x-k_1\Omega_{0x})^{r_1}}\frac{\sin^{r_2}\left( \frac{ \Omega_y-k_2\Omega_{0y}}{\Omega_{0y}}\pi \right)}{(\Omega_y-k_2\Omega_{0y})^{r_2}}, \nonumber\\
=& \left(\frac{\pi^{(r_1 + r2)}}{\Omega_{0x}^{r_1}\Omega_{0y}^{r_2}}\right)\sum_{k_1\in\mathcal{K}_1}\sum_{k_2\in\mathcal{K}_2}\sum_{m_1=1}^{r_1}\sum_{m_2=1}^{r_2} q_{k_1,k_2,m_1,m_2} (-1)^{r_1k_1 + r_2k_2} \sum_{c_1=0}^{r_1-m_1} {r_1-m_1 \choose c_1} (-k_1\Omega_{0x})^{r_1-m_1-c_1}\Omega_x^{c_1} \nonumber \\  & \text{sinc}^{r_1}\left( \frac{ \Omega_x-k_1\Omega_{0x}}{\Omega_{0x}} \right) \sum_{c_2=0}^{r_2-m_2} {r_2-m_2 \choose c_2} (-k_2\Omega_{0y})^{r_2-m_2-c_2}\Omega_y^{c_2} \text{sinc}^{r_2}\left( \frac{ \Omega_y-k_2\Omega_{0y}}{\Omega_{0y}} \right) . 
\label{eq:G1}
\end{align}
Now, applying the inverse Fourier transform on (\ref{eq:G1}), $g_S(x,y)$ is obtained as
\begin{align}
g_S(x,y)=& \left(\frac{\pi^{(r_1 + r2)}}{\Omega_{0x}^{r_1}\Omega_{0y}^{r_2}}\right) \sum_{k_1\in\mathcal{K}_1}\sum_{k_2\in\mathcal{K}_2}\sum_{m_1=1}^{r_1}\sum_{m_2=1}^{r_2} q_{k_1,k_2,m_1,m_2} (-1)^{r_1k_1 + r_2k_2} \sum_{c_1=0}^{r_1-m_1} {r_1-m_1 \choose c_1} (-k_1\Omega_{0x})^{r_1-m_1-c_1} \nonumber\\
&{\frac{1}{\mathrm{j}^{c_1}} \frac{\text{d}^{c_1}}{\text{d}x^{c_1}}\left[e^{\mathrm{j}k_1\Omega_{0x}x} \frac{1}{T_{0x}}\beta^{r_1-1} \left(\frac{x}{T_{0x}}\right) \right]} \sum_{c_2=0}^{r_2-m_2} {r_2-m_2 \choose c_2} (-k_2\Omega_{0y})^{r_2-m_2-c_2} {\frac{1}{\mathrm{j}^{c_2}} \frac{\text{d}^{c_2}}{\text{d}y^{c_2}}\left[e^{\mathrm{j}k_2\Omega_{0y}y} \frac{1}{T_{0y}}\beta^{r_2-1} \left(\frac{y}{T_{0y}}\right) \right]},
\label{eq:g_x_y_sepa}
\end{align}
where $\beta^{r}(t)$ denotes the $r^{\text{th}}$-order polynomial B-spline. 

Now, evaluating the above equation for $V(\mathrm{j}\Omega_{0x},\mathrm{j}\Omega_{0y})$ given in (\ref{eq:separable_V}), we arrive at the class of separable 2-D SMS kernels given in (\ref{eq:2-D_SMS}). 


\section{Proof of Proposition \ref{prop:non_separable}}
\label{appendix:prop_nonseparable}

The Proof of Proposition \ref{prop:non_separable} involves two steps: (i) showing $G_{NS}(s_x,s_y)$ is an EFET; and (ii) showing $G_{NS}(\mathrm{j}\Omega_x, \mathrm{j}\Omega_y) \in L^2(\mathbb{R})$. 
\subsection{Proof that $G_{NS}(s_x,s_y)$ is an EFET }
\label{appendix:G_efet_nonseparable}
\begin{proof}[Proof]

For the nonseparable case, we have 
\begin{equation}
    Z_{r_1,r_2}(s_x,s_y) = \sinh\left(\frac{\pi s_x}{\Omega_{0x}} + \frac{\pi s_y}{\Omega_{0y}}\right)\sinh\left(\frac{\pi s_y}{\Omega_{0y}} - \frac{\pi s_x}{\Omega_{0x}}\right). \nonumber
\end{equation}

We know that
\begin{align}
    \left|\sinh\left(\frac{\pi s_x}{\Omega_{0x}} + \frac{\pi s_y}{\Omega_{0y}}\right)\sinh\left(\frac{\pi s_y}{\Omega_{0y}} - \frac{\pi s_x}{\Omega_{0x}}\right)\right| \leq& \left|\sinh\left(\frac{\pi s_x}{\Omega_{0x}} + \frac{\pi s_y}{\Omega_{0y}}\right)\right| \left|\sinh\left(\frac{\pi s_y}{\Omega_{0y}} - \frac{\pi s_x}{\Omega_{0x}}\right)\right|, \nonumber \\ 
    =&\, \left| \frac{e^{\left(\frac{\pi s_x}{\Omega_{0x}} + \frac{\pi s_y}{\Omega_{0y}}\right)} - e^{ - \left(\frac{\pi s_x}{\Omega_{0x}} + \frac{\pi s_y}{\Omega_{0y}}\right)} }{2} \right| \left|\frac{e^{\left(\frac{\pi s_y}{\Omega_{0y}} - \frac{\pi s_x}{\Omega_{0x}}\right)} - e^{ - \left(\frac{\pi s_y}{\Omega_{0y}} - \frac{\pi s_x}{\Omega_{0x}}\right)} }{2} \right|, \nonumber \\ 
    \leq&\, \frac{1}{4} \left( \left| e^{\left(\frac{\pi s_x}{\Omega_{0x}} + \frac{\pi s_y}{\Omega_{0y}}\right)}\right| + \left|e^{ - \left(\frac{\pi s_x}{\Omega_{0x}} + \frac{\pi s_y}{\Omega_{0y}}\right)} \right| \right) \left(\left|e^{\left(\frac{\pi s_y}{\Omega_{0y}} - \frac{\pi s_x}{\Omega_{0x}}\right)} \right| + \left| e^{-\left(\frac{\pi s_y}{\Omega_{0y}} - \frac{\pi s_x}{\Omega_{0x}}\right)} \right|\right), \nonumber \\
    \leq& \,\frac{1}{4} \left(  e^{\left|\frac{\pi s_x}{\Omega_{0x}} + \frac{\pi s_y}{\Omega_{0y}}\right|} + e^{\left|\frac{\pi s_x}{\Omega_{0x}} + \frac{\pi s_y}{\Omega_{0y}}\right|}  \right) \left(e^{\left|\frac{\pi s_y}{\Omega_{0y}} - \frac{\pi s_x}{\Omega_{0x}}\right|}  +  e^{\left|\frac{\pi s_y}{\Omega_{0y}} - \frac{\pi s_x}{\Omega_{0x}}\right|} \right), \nonumber \\
    =& \,\frac{1}{4} \left( 2\, e^{\pi\left|\frac{ s_x}{\Omega_{0x}} + \frac{ s_y}{\Omega_{0y}}\right|}\right) \left( 2\, e^{\pi\left|\frac{ s_y}{\Omega_{0y}} - \frac{ s_x}{\Omega_{0x}}\right|}\right), \nonumber \\
    \le& \, \left(e^{\pi \left(\left|\frac{s_x}{\Omega_{0x}}\right| + \left|\frac{s_y}{\Omega_{0y}}\right|\right)} \right)\left( e^{\pi \left(\left|\frac{s_y}{\Omega_{0y}}\right| + \left|\frac{s_x}{\Omega_{0x}}\right|\right)} \right), \nonumber \\
    =& \, e^{\left(\frac{2\pi}{\Omega_{0x}}\left|s_x \right|  + \frac{2\pi}{\Omega_{0y}}\left|s_y\right|\right)} . \nonumber
\end{align}

Since $Z(\mathrm{j}\Omega_x, \mathrm{j}\Omega_y)$  is a sum of entire functions, it is an EFET, and consequently $G_{NS}(\mathrm{j}\Omega_x, \mathrm{j}\Omega_y)$ is an EFET as well. The support of $g_{NS}(x,y)$ is given by $\left[-\frac{2\pi}{\Omega_{0x}}, \frac{2\pi}{\Omega_{0x}} \right] \times \left[-\frac{2\pi}{\Omega_{0y}}, \frac{2\pi}{\Omega_{0y}} \right] $ i.e., $\left[-T_{0x}, T_{0x} \right] \times \left[-T_{0y}, T_{0y} \right] $.

\end{proof}


\subsection{Proof that $G_{NS}(\mathrm{j}\Omega_x, \mathrm{j}\Omega_y) \in L^2(\mathbb{R})$}
\label{appendix:G_energy_nonseparable}
\begin{proof}[Proof]

Using the hyperbolic sine identity: $\sinh(\mathrm{j}\theta)=\mathrm{j}\sin(\theta)$, $\forall \,\theta\in \mathbb{R}$, $G_{NS}(\mathrm{j}\Omega_{0x}, \mathrm{j}\Omega_{0y})$ is expressed as
\begin{align}
G_{NS}(\mathrm{j}\Omega_x, \mathrm{j}\Omega_y)
= \frac{\mathrm{j}^{(2 - |\mathcal{K}_1| - |\mathcal{K}_2|)}\,V(\mathrm{j}\Omega_x, \mathrm{j}\Omega_y)\,\sin\left(\frac{\pi\,\Omega_x}{\Omega_{0x}}\,+\, \frac{\pi\,\Omega_y}{\Omega_{0y}}\right)\sin\left(\frac{\pi\,\Omega_y}{\Omega_{0y}}\,-\, \frac{\pi\,\Omega_x}{\Omega_{0x}}\right)}{\prod \limits_{\substack{k_1\in\mathcal{K}_1 \\ k_2\in\mathcal{K}_2}}\left(\frac{\Omega_x}{\Omega_{0x}} + \frac{\Omega_y}{\Omega_{0y}} - k_1 - k_2 \right)\left(\frac{\Omega_y}{\Omega_{0y}} - \frac{\Omega_x}{\Omega_{0x}} - k_2 + k_1 \right)}, \nonumber
\end{align}

which when expanded using partial fraction decomposition takes the form
\begin{equation}
\centering
G_{NS}(\mathrm{j}\Omega_x, \mathrm{j}\Omega_y)=\sum_{k_1\in\mathcal{K}_1}\sum_{k_2\in\mathcal{K}_2}\frac{q_{k_1,k_2}}{\left(\frac{\Omega_x}{\Omega_{0x}} + \frac{\Omega_y}{\Omega_{0y}} - k_1 - k_2 \right)\left(\frac{\Omega_y}{\Omega_{0y}} - \frac{\Omega_x}{\Omega_{0x}} - k_2 + k_1 \right)}\sin\left(\frac{\pi\,\Omega_x}{\Omega_{0x}}\,+\, \frac{\pi\,\Omega_y}{\Omega_{0y}}\right)\sin\left(\frac{\pi\,\Omega_y}{\Omega_{0y}}\,-\, \frac{\pi\,\Omega_x}{\Omega_{0x}}\right).
\label{eq:non_sepa_partial}
\end{equation}
$q_{k_1,k_2}$ are the coefficients of the partial fraction expansion, and are given by
\begin{equation}
\centering 
q_{k_1,k_2}= \left.\left( \frac{\mathrm{j}^{(2 - |\mathcal{K}_1| - |\mathcal{K}_2|)}\,V(\mathrm{j}\Omega_x, \mathrm{j}\Omega_y)}{\prod \limits_{\substack{p_1\in\mathcal{K}_1\backslash\{k_1\} \\ p_2\in\mathcal{K}_2\backslash\{k_2\}}}\left(\frac{\Omega_x}{\Omega_{0x}} + \frac{\Omega_y}{\Omega_{0y}} - p_1 - p_2 \right)\left(\frac{\Omega_y}{\Omega_{0y}} - \frac{\Omega_x}{\Omega_{0x}} - p_2 + p_1 \right)} \right ) \right|_{\substack{\Omega_x=k_1\Omega_{0x} \\ \Omega_y=k_2\Omega_{0y}}}.\nonumber
\end{equation}

Using the Cauchy-Schwarts inequality, we have
\begin{align}
    \int\limits_{-\infty}^{\infty}\int\limits_{-\infty}^{\infty}|G_{NS}(\mathrm{j}\Omega_x, \mathrm{j}\Omega_y)|^2 \text{d}\Omega_x\text{d}\Omega_y = & \int\limits_{-\infty}^{\infty}\int\limits_{-\infty}^{\infty} \left| \sum_{k_1\in\mathcal{K}_1}\sum_{k_2\in\mathcal{K}_2}\,q_{k_1,k_2}\,\frac{\sin\left(\frac{\pi\,\Omega_x}{\Omega_{0x}}\,+\, \frac{\pi\,\Omega_y}{\Omega_{0y}}\right)\sin\left(\frac{\pi\,\Omega_y}{\Omega_{0y}}\,-\, \frac{\pi\,\Omega_x}{\Omega_{0x}}\right)}{\left(\frac{\Omega_x - k_1\,\Omega_{0x}}{\Omega_{0x}} + \frac{\Omega_y - k_2\,\Omega_{0y}}{\Omega_{0y}}\right)\left(\frac{\Omega_y - k_2\,\Omega_{0y}}{\Omega_{0y}} - \frac{\Omega_x - k_1\,\Omega_{0x}}{\Omega_{0x}}\right)}\right|^2\text{d}\Omega_x\text{d}\Omega_y, \nonumber \\
    \leq& |\mathcal{K}_1||\mathcal{K}_2|\sum_{k_1\in\mathcal{K}_1}\sum_{k_2\in\mathcal{K}_2}\,|q_{k_1,k_2}|^2 \nonumber \\ & \int\limits_{-\infty}^{\infty}\int\limits_{-\infty}^{\infty}\left|\frac{\sin\left(\frac{\pi\,\Omega_x}{\Omega_{0x}}\,+\, \frac{\pi\,\Omega_y}{\Omega_{0y}}\right)\sin\left(\frac{\pi\,\Omega_y}{\Omega_{0y}}\,-\, \frac{\pi\,\Omega_x}{\Omega_{0x}}\right)}{\left(\frac{\Omega_x - k_1\,\Omega_{0x}}{\Omega_{0x}} + \frac{\Omega_y - k_2\,\Omega_{0y}}{\Omega_{0y}}\right)\left(\frac{\Omega_y - k_2\,\Omega_{0y}}{\Omega_{0y}} - \frac{\Omega_x - k_1\,\Omega_{0x}}{\Omega_{0x}}\right)}\right|^2 \text{d}\Omega_x\text{d}\Omega_y . \nonumber
\end{align}

Let us consider
\begin{align}
    \zeta = \int\limits_{-\infty}^{\infty}\int\limits_{-\infty}^{\infty}\left|\frac{\sin\left(\frac{\pi\,\Omega_x}{\Omega_{0x}}\,+\, \frac{\pi\,\Omega_y}{\Omega_{0y}}\right)\sin\left(\frac{\pi\,\Omega_y}{\Omega_{0y}}\,-\, \frac{\pi\,\Omega_x}{\Omega_{0x}}\right)}{\left(\frac{\Omega_x - k_1\,\Omega_{0x}}{\Omega_{0x}} + \frac{\Omega_y - k_2\,\Omega_{0y}}{\Omega_{0y}}\right)\left(\frac{\Omega_y - k_2\,\Omega_{0y}}{\Omega_{0y}} - \frac{\Omega_x - k_1\,\Omega_{0x}}{\Omega_{0x}}\right)}\right|^2 \text{d}\Omega_x\text{d}\Omega_y .
    \label{eq:double_int} 
\end{align}

Substituting $\widetilde{\Omega}_x = \frac{\Omega_x - k_1\,\Omega_{0x}}{\Omega_{0x}}$ and $ \widetilde{\Omega}_y = \frac{\Omega_y - k_2\,\Omega_{0y}}{\Omega_{0y}}  $ in (\ref{eq:double_int}) and simplifying it further, we get 

\begin{align}
    \zeta =& \,\,\Omega_{0x}\Omega_{0y} \int\limits_{-\infty}^{\infty}\int\limits_{-\infty}^{\infty}\left|\frac{\sin\left(\pi\,\widetilde{\Omega}_x + k_1\pi\,+\, \pi\,\widetilde{\Omega}_y + k_2\pi \right)\sin\left(\pi\,\widetilde{\Omega}_y + k_2\pi\,-\, \pi\,\widetilde{\Omega}_x + k_1\pi\right)}{\left(\widetilde{\Omega}_x + \widetilde{\Omega}_y\right)\left(\widetilde{\Omega}_x + \widetilde{\Omega}_y\right)}\right|^2 \text{d}\widetilde{\Omega}_x\text{d}\widetilde{\Omega}_y, \nonumber \\
    =& \,\,\Omega_{0x}\Omega_{0y} \int\limits_{-\infty}^{\infty}\int\limits_{-\infty}^{\infty}\left|\frac{\sin\left(\pi\,(\widetilde{\Omega}_x + \widetilde{\Omega}_y) \right)\sin\left(\pi\,(\widetilde{\Omega}_y - \widetilde{\Omega}_x)\right)}{\left(\widetilde{\Omega}_x + \widetilde{\Omega}_y\right)\left(\widetilde{\Omega}_x + \widetilde{\Omega}_y\right)}\right|^2 \text{d}\widetilde{\Omega}_x\text{d}\widetilde{\Omega}_y, \nonumber \\
    =& \,\,\pi^4\Omega_{0x}\Omega_{0y} \int\limits_{-\infty}^{\infty}\int\limits_{-\infty}^{\infty}\left|\text{sinc}\left(\widetilde{\Omega}_x + \widetilde{\Omega}_y \right)\text{sinc}\left(\widetilde{\Omega}_y - \widetilde{\Omega}_x\right)\right|^2 \text{d}\widetilde{\Omega}_x\text{d}\widetilde{\Omega}_y. \label{eq:before_rot}
\end{align}

Consider rotation of a function $\vartheta(\widetilde{\Omega}_x,\widetilde{\Omega}_y) = \text{sinc}\left(\widetilde{\Omega}_x + \widetilde{\Omega}_y \right)\text{sinc}\left(\widetilde{\Omega}_y - \widetilde{\Omega}_x\right)$  by $\theta = 45^{\circ}$. Since energy of $\vartheta(\widetilde{\Omega}_x,\widetilde{\Omega}_y)$ is preserved by the rotation operator $R = \left[\begin{smallmatrix}\cos\theta&-\sin\theta\\\sin\theta&\cos\theta\end{smallmatrix}\right]$, we substitute the rotated $\vartheta(\widetilde{\Omega}_x,\widetilde{\Omega}_y)$ in \eqref{eq:before_rot}. Thus the right-hand side of (\ref{eq:before_rot}) is given by 
\begin{align}
    \zeta =&\,\,\pi^4\Omega_{0x}\Omega_{0y} \int\limits_{-\infty}^{\infty}\int\limits_{-\infty}^{\infty}\left|\text{sinc}\left(\sqrt{2}\widetilde{\Omega}_x \right)\text{sinc}\left(\sqrt{2}\widetilde{\Omega}_y\right)\right|^2 \text{d}\widetilde{\Omega}_x\text{d}\widetilde{\Omega}_y, \nonumber \\
    =& \,\,\pi^4\Omega_{0x}\Omega_{0y} \int\limits_{-\infty}^{\infty}\left|\text{sinc}\left(\sqrt{2}\widetilde{\Omega}_x \right)\right|^2\text{d}\widetilde{\Omega}_x\int\limits_{-\infty}^{\infty}\left| \text{sinc}\left(\sqrt{2}\widetilde{\Omega}_y\right)\right|^2 \text{d}\widetilde{\Omega}_y,\nonumber\\
    =& \,\,\pi^4\Omega_{0x}\Omega_{0y}\left(\frac{1}{\sqrt{2}}\right)\left( \frac{1}{\sqrt{2}}\right),\nonumber\\
    =& \,\,\frac{\pi^4\Omega_{0x}\Omega_{0y}}{2}.\nonumber
\end{align}
Hence, we have $\zeta < \infty $, and consequently, $\int\limits_{-\infty}^{\infty}\int\limits_{-\infty}^{\infty}|G_{NS}(\mathrm{j}\Omega_x, \mathrm{j}\Omega_y)|^2 \text{d}\Omega_x\text{d}\Omega_y < \infty $.

\end{proof}

\subsection{Expression for $g_{NS}(x,y)$}
\label{appendix:kernel_nonseparable_expression}
Consider
\begin{equation}
\centering
G_{NS}(\mathrm{j}\Omega_x, \mathrm{j}\Omega_y)=\sum_{k_1\in\mathcal{K}_1}\sum_{k_2\in\mathcal{K}_2}\frac{q_{k_1,k_2}}{\left(\frac{\Omega_x}{\Omega_{0x}} + \frac{\Omega_y}{\Omega_{0y}} - k_1 - k_2 \right)\left(\frac{\Omega_y}{\Omega_{0y}} - \frac{\Omega_x}{\Omega_{0x}} - k_2 + k_1 \right)}\sin\left(\frac{\pi\,\Omega_x}{\Omega_{0x}}\,+\, \frac{\pi\,\Omega_y}{\Omega_{0y}}\right)\sin\left(\frac{\pi\,\Omega_y}{\Omega_{0y}}\,-\, \frac{\pi\,\Omega_x}{\Omega_{0x}}\right).
\end{equation}
Upon simplifying further, we get
\begin{align}
G_{NS}(\mathrm{j}\Omega_{0x},\mathrm{j}\Omega_{0y}) =& \sum_{k_1\in\mathcal{K}_1}\sum_{k_2\in\mathcal{K}_2}\,q_{k_1,k_2}\,\frac{\sin\left(\frac{\pi\,\Omega_x}{\Omega_{0x}}\,+\, \frac{\pi\,\Omega_y}{\Omega_{0y}}\right)\sin\left(\frac{\pi\,\Omega_y}{\Omega_{0y}}\,-\, \frac{\pi\,\Omega_x}{\Omega_{0x}}\right)}{\left(\frac{\Omega_x - k_1\,\Omega_{0x}}{\Omega_{0x}} + \frac{\Omega_y - k_2\,\Omega_{0y}}{\Omega_{0y}}\right)\left(\frac{\Omega_y - k_2\,\Omega_{0y}}{\Omega_{0y}} - \frac{\Omega_x - k_1\,\Omega_{0x}}{\Omega_{0x}}\right)}, \nonumber\\
=&  \sum_{k_1\in\mathcal{K}_1}\sum_{k_2\in\mathcal{K}_2}\,q_{k_1,k_2}\,(-1)^{(k_1 + k_2 + k_2 - k_1)}\frac{\sin\left(\frac{\pi\,\Omega_x}{\Omega_{0x}}\,+\, \frac{\pi\,\Omega_y}{\Omega_{0y}}  - \pi(k_1 + k_2)\right)\sin\left(\frac{\pi\,\Omega_y}{\Omega_{0y}}\,-\, \frac{\pi\,\Omega_x}{\Omega_{0x}}  - \pi(k_2 - k_1)\right)}{\left(\frac{\Omega_x - k_1\,\Omega_{0x}}{\Omega_{0x}} + \frac{\Omega_y - k_2\,\Omega_{0y}}{\Omega_{0y}}\right)\left(\frac{\Omega_y - k_2\,\Omega_{0y}}{\Omega_{0y}} - \frac{\Omega_x - k_1\,\Omega_{0x}}{\Omega_{0x}}\right)}, \nonumber\\
=&  \,\, \pi^2 \sum_{k_1\in\mathcal{K}_1}\sum_{k_2\in\mathcal{K}_2}\,q_{k_1,k_2}\,{\text{sinc}\left(\frac{\Omega_x - k_1\,\Omega_{0x}}{\Omega_{0x}}\,+\, \frac{\Omega_y - k_2\,\Omega_{0y}}{\Omega_{0y}}\right)\text{sinc}\left(\frac{\Omega_y - k_2\,\Omega_{0y}}{\Omega_{0y}}\,-\, \frac{\Omega_x - k_1\,\Omega_{0x}}{\Omega_{0x}}\right)}.
\label{eq:G_non_sepa}
\end{align}
Using the Lemma \ref{lemma:inverse_fourier_nonseparable_sinc} in Appendix \ref{appendix:add_proposition} and the frequency shift property, the inverse Fourier transform of (\ref{eq:G_non_sepa}) is given by
\begin{align}
g_{NS}(x,y)=&\,\,\,\,\frac{\Omega_{0x}\Omega_{0y}}{8}\,\,\text{rect}\left( \frac{\Omega_{0x}x + \Omega_{0y}y}{4\pi}\right)\text{rect}\left( \frac{\Omega_{0y}y - \Omega_{0x}x}{4\pi}\right)\, \sum_{k_1\in\mathcal{K}_1}\sum_{k_2\in\mathcal{K}_2}q_{k_1,k_2}\, e^{\mathrm{j}(k_1\Omega_{0x}x + k_2\Omega_{0y}y)}.
\label{eq:g_x_y_non_sepa}
\end{align}


\section{ }
\label{appendix:add_proposition}

\begin{lemma}
\label{lemma:inverse_fourier_nonseparable_sinc}
If $F(\Omega_x, \Omega_y) = \text{sinc}\left(\frac{\Omega_x}{\Omega_{0x}} + \frac{\Omega_y}{\Omega_{0y}}\right)\text{sinc}\left(\frac{\Omega_y}{\Omega_{0y}} - \frac{\Omega_x}{\Omega_{0x}}\right)$, then its inverse Fourier transform is given by $f(x,y) = \,\,\frac{\Omega_{0x}\Omega_{0y}}{8\pi^2}\,\,\text{rect}\left( \frac{\Omega_{0x}x + \Omega_{0y}y}{4\pi}\right)\text{rect}\left( \frac{\Omega_{0y}y - \Omega_{0x}x}{4\pi}\right)$.
\end{lemma}
\begin{proof}[Proof]

We know that the 1-D inverse Fourier transform of $F(\Omega) = \text{sinc}(\sqrt{2}\Omega)$  is given by

\begin{equation}
    f(t) = \frac{1}{2\sqrt{2}\pi}\text{rect}\left(\frac{t}{2\sqrt{2}\pi}\right).\nonumber
\end{equation}

For the separable two-dimensional counterpart, the inverse Fourier transform of the function
\begin{equation}
    \widetilde{F}(\Omega_x, \Omega_y) = \text{sinc}(\sqrt{2}\Omega_x)\text{sinc}(\sqrt{2}\Omega_y), \label{eq:foo}
\end{equation}
 is given by,
\begin{equation}
    \widetilde{f}(x,y) = \frac{1}{8\pi^2}\text{rect}\left(\frac{x}{2\sqrt{2}\pi}\right)\text{rect}\left(\frac{y}{2\sqrt{2}\pi}\right). \nonumber
\end{equation}

Applying rotation of $\theta = - 45^{\circ}$ on the function in (\ref{eq:foo}) and using the 2-D rotation property of the Fourier transform we have,
\begin{align}
F\left(\frac{\Omega_x}{\sqrt{2}} + \frac{\Omega_y}{\sqrt{2}}, - \frac{\Omega_x}{\sqrt{2}} + \frac{\Omega_y}{\sqrt{2}}\right) 
=&\,\,\text{sinc}\left({\Omega_x} + {\Omega_y}\right)\text{sinc}\left({\Omega_y} - {\Omega_x}\right),\nonumber
\end{align}
and its corresponding inverse Fourier transform is given by
\begin{align}
\hspace{-.6in} f(x,y) =& \,\,\frac{1}{8\pi^2}\text{rect}\left(\frac{\frac{x}{\sqrt{2}} + \frac{y}{\sqrt{2}}}{2\sqrt{2}\pi}\right)\text{rect}\left(\frac{\frac{y}{\sqrt{2}} - \frac{x}{\sqrt{2}}}{2\sqrt{2}\pi}\right),\nonumber\\
\hspace{0.1in}=& \,\,\frac{1}{8\pi^2}\text{rect}\left(\frac{x + y}{4\pi}\right)\text{rect}\left(\frac{y - x}{4\pi}\right).
\label{eq:goo}
\end{align}

Using the two dimensional scaling property on (\ref{eq:goo}), the inverse Fourier transform of 

\begin{equation}
    F(\Omega_x, \Omega_y) = \text{sinc}\left(\frac{\Omega_x}{\Omega_{0x}} + \frac{\Omega_y}{\Omega_{0y}}\right)\text{sinc}\left(\frac{\Omega_y}{\Omega_{0y}} - \frac{\Omega_x}{\Omega_{0x}}\right),\nonumber
\end{equation}
is given by,
\begin{equation}
    f(x,y) = \,\,\frac{\Omega_{0x}\Omega_{0y}}{8\pi^2}\,\,\text{rect}\left( \frac{\Omega_{0x}x + \Omega_{0y}y}{4\pi}\right)\text{rect}\left( \frac{\Omega_{0y}y - \Omega_{0x}x}{4\pi}\right).\nonumber
\end{equation}

\end{proof}





\bibliographystyle{IEEEbib}
\bibliography{references}

\end{document}